\documentclass[11pt]{article}
\usepackage{amssymb}
\usepackage{enumerate}
\usepackage{enumitem}
\usepackage{amsmath,amsthm,marvosym,wasysym,mathdots}
\usepackage{bbm}
\usepackage{natbib}
\usepackage{color}
\usepackage{setspace}
\usepackage{epsfig}
\usepackage{hyperref}
\hypersetup{%
  colorlinks=false,
  pdfborder = {0 0 0.5 [3 0]}
}
\usepackage{breakurl} 
\usepackage{multirow}
\usepackage[left=1in,right=1in,top=1.1in,bottom=1in]{geometry}

\newcommand{\indic}[1]{1\hspace{-2.1mm}{1}_{\{#1\}}} 

\DeclareMathOperator*{\esssup}{ess\,sup}

\newtheorem{theorem}{Theorem}
\newtheorem*{theorem*}{Theorem}

\newtheorem{corollary}[theorem]{Corollary}

\newtheorem{definition}[theorem]{Definition}

\newtheorem{proposition}[theorem]{Proposition}
\newtheorem{remark}[theorem]{Remark}

\def\N{{\mathbb N}} 
\def\Q{{\mathbb Q}} 
\def\R{{\mathbb R}} 
\def\P{{\mathbb P}} 
\newcommand{\EE}{{\mathord{I\kern -.33em E}}}
\def\E{{\mathbb E}} 
\def\Fil{{\mathbb F}} 

\def\1{1{\hskip -3.3 pt}\hbox{I}}
\def\F{{\mathcal F}}

\providecommand{\varitem}{}


\numberwithin{equation}{section}
\numberwithin{theorem}{section}

\begin{document}

\title{Optimal Derivative Liquidation Timing Under Path-Dependent Risk Penalties}
\author{Tim Leung\thanks{ \small{Corresponding author. Industrial Engineering \& Operations Research Department,
Columbia University, New York, NY 10027. Tel: (212) 854-2942. E-mail:
\mbox{leung@\,ieor.columbia.edu}.} } \and Yoshihiro Shirai\thanks{\small{Industrial
Engineering \& Operations Research Department, Columbia University,
New York, NY 10027. Tel: (212) 854-2942. E-mail:
\mbox{yoshihiroshirai@gmail.com}. }} }\date{\today} \maketitle

\begin{small}

\begin{abstract}This paper studies the risk-adjusted optimal timing to liquidate an option at the prevailing market price. In addition to maximizing the expected discounted return from option sale, we incorporate a path-dependent risk penalty based on shortfall or quadratic variation of the option price  up to the liquidation time. We establish the conditions under which it is optimal to immediately liquidate or hold the option position through expiration. Furthermore, we study the variational inequality associated with the optimal stopping problem, and prove the existence and uniqueness of a strong solution. A series of analytical and numerical results  are provided to illustrate the non-trivial optimal liquidation strategies under   geometric Brownian motion (GBM) and exponential Ornstein-Uhlenbeck models. We examine the combined effects of price dynamics and risk penalty on the sell and delay regions for  various options.  In addition, we obtain an explicit closed-form solution  for the liquidation of a stock with quadratic penalty under the GBM model.   
\end{abstract}

\tableofcontents

\end{small}

\section{Introduction}
For decades, options have been widely used as a tool for investment and risk management. As of 2012, the daily market notional for S\&P 500 options   is about US\$90 billion and the average daily volume has grown rapidly from 119,808 in 2002 to 839,108 as of Jan 2013\footnote{See \url{http://www.cboe.com/micro/spx/introduction.aspx}}. Empirical studies on options returns often assume that the options  are held to maturity (see \cite{Broadie2009} and   references therein). For every liquidly traded option, there is an embedded timing flexibility to  liquidate the position through the market prior to expiry.   Hence, an important question  for effective risk management is: when is the best time to sell an option? In this paper, we propose a risk-adjusted optimal stopping framework to address this problem for a variety of options under different underlying price dynamics.

In addition to  maximizing the expected discounted market value to be received from option sale, we incorporate a risk penalty that accounts for adverse price movements till the liquidation time. For every candidate strategy, we measure the associated risk by integrating  over time the realized shortfall, or more generally its transformation in terms of a loss function, of the option position. As such, our  integrated shortfall risk penalty is path dependent and introduces the trade-off between risk and return for every liquidation timing strategy.   

Under a general diffusion model for the underlying stock price, we formulate an optimal stopping problem that includes an integral penalization term. To this end, we define and apply the concept of optimal liquidation premium which represents the additional value from optimally waiting to sell, as opposed to immediate liquidation. As it turns out, it is optimal for the option holder to sell as soon as this premium vanishes. This observation leads to  a number of useful mathematical characterizations and financial interpretations of the optimal liquidation strategies for various positions. 

We first identify the conditions under which it is optimal to immediately liquidate or hold the option position through expiration. The investigation of the non-trivial liquidation strategies involves the analytical and numerical studies of the inhomogeneous  variational inequality associated with the optimal stopping problem. In a related work,  \cite{budhi12}  examines  the solution structure for a   finite maturity optimal stopping problem under L\'{e}vy processes with a running cost and other features, and an inhomogeneous  variational inequality also arises from the associated partial integro-differential free-boundary problem. In the context of asset management,   \cite{egamiasset} study a perpetual optimal stopping problem with a running cash flow generated from dividend and
coupon payments, and they solve a time-independent inhomogeneous variational inequality.  For the  variational inequalities in our liquidation problems, we   prove the existence and uniqueness of a strong solution \`a la   \cite{Bensoussan78} (see Section  \ref{Sect-ExUn} below)  under general conditions applicable to  both geometric Brownian motion (GBM) (see \cite{Merton1973}) and exponential Ornstein-Uhlenbeck (OU) (see \cite{Ornstein1930}) models for the underlying dynamics. We also provide some mathematical characterizations and numerical examples of the optimal liquidation strategies for stocks,  calls, puts, and  straddles.

The incorporation of the risk penalty gives rise to optimal liquidation strategies that are distinctly different from the unpenalized case.  For instance, if the option's Delta (derivative of the option price with respect to the underlying price) is of the same constant sign as the excess return of the underlying, then it is optimal to hold the option till maturity when there is no risk penalty
(see Prop. \ref{trivial_timing}). This applies to the case of a call option (resp. a put option) if the investor is bullish (resp. bearish) on the stock. However, under  risk penalization the investor  may  find it optimal  to liquidate a call (resp. a put) early even in the bullish (resp. bearish) scenario (see e.g.  Prop. \ref{Put4}). Furthermore, the shape of the optimal liquidation region depends significantly  on the risk penalty.  We show that  higher risk penalization coefficient always reduces the delay  region, which intuitively means that the investor is more likely to sell earlier. Moreover, in some cases the optimal delay and sell regions can exhibit some interesting  structures, such as  disconnectedness (see Figures \ref{Put_SF} and \ref{Call_OU}). These analytical results are significantly facilitated by  the properties of the optimal liquidation premium (see Theorems \ref{prop1_L_G} and \ref{compactsupport}).

Our path-dependent risk penalization model can also be viewed as an alternative way to incorporate the investor's risk sensitivity in option  liquidation/exercise timing problems, as compared to   the utility maximization/indifference pricing approach \citep{henderson2011optimal, LeungLudkovski2, LSZ12}. On the other hand,    \cite{LeungLudkovski2011} investigate  the optimal timing to buy equity European and American options without risk penalty under incomplete markets, where the investor is assumed to select risk-neutral pricing measure  different from the market's.   \cite{LeungLiubookchap} also  discuss the timing to sell an option under the GBM model without any risk penalty, which is a special example of  our  model.

As is well known, the  concept of risk  measures based on  shortfall risk has been applied to many  portfolio optimization problems; see  \cite{Artzner,Rockafellar, FS02, FScB, sircarRM},  and references therein.  Our model applies  this idea to options trading  as a path-penalty  associated with each  liquidation strategy.  As a variation of the shortfall we also introduce a risk penalty based on the quadratic variation of option price process. In particular, we obtain an explicit closed-form solution for the liquidation of a stock with quadratic penalty under the GBM model (see Theorem \ref{perpL}). Through examining the optimal liquidation premium, we also compare the liquidation strategies for calls and puts under the shortfall-based and quadratic  risk penalties.    \cite{Forsyth2012} also adopt the mean-quadratic-variation as a criterion for determining the optimal stock trading strategy in the presence of price impact.  

The recent paper by \cite{ziemba}  considers a discrete-time portfolio optimization problem with a convex loss function that accounts for the shortfall  of  the wealth  trajectory from a benchmark.  While we consider the problem of optimal liquidation of stocks and options, their investigation  focuses on  the optimal capital growth or Kelly strategy. On the other hand, \cite{frei13}  study the optimal liquidation of a stock position subject to temporary price impact. Specifically, they  minimize  the mean and variance of the order slippage with respect to   the VWAP
(volume weighted average price) as the benchmark. These papers adopt the stochastic control approach to solve for the optimal position over time, whereas our problems concern only the optimal timing to liquidate.

The rest of the paper is organized as follows. In Section \ref{sec:ProbForm}, we formulate the optimal liquidation problem for a generic European claim  in a   diffusion market. In subsequent sections, we focus on the liquidation of a stock or an option  under the GBM and exponential OU models. In Sections \ref{sect-GBM} and \ref{sect-OU}, we study the optimal liquidation timing with  a shortfall risk penalty. In Section \ref{sect:quadr}, we conduct our analysis with a  quadratic variation risk penalty. Section \ref{sect:concl} concludes the paper. In Section \ref{Sect-ExUn},  we    discuss    the existence of a strong solution to the variational inequality as well as the probabilistic representation satisfied by the optimal liquidation premium.

\section{Problem Overview}\label{sec:ProbForm}
In the background, we fix a  probability space $(\Omega, \F, \P)$, where $\P$ is the historical probability measure. The market consists of a risky asset  $S$ and a money market account with a constant positive interest rate $r$. The risky asset price is modeled by a positive diffusion process following the stochastic differential equation
\begin{equation}\label{SDE_S}
dS_t =\mu(t,S_t)S_tdt + \sigma(t,S_t)S_tdW_t, \qquad S_0=s,
\end{equation}
where  $W$ is a standard Brownian motion under measure $\P$ and $s>0$. Here, the deterministic  coefficients $\mu(t,s)$ and $\sigma(t,s)$ are assumed to satisfy standard Lipschitz and growth conditions \cite[\S 5.2]{KaratzasShreve91} to ensure a unique strong solution to \eqref{SDE_S}. We let    $\mathbb{F}=(\F_t)_{t\geq 0}$ be the filtration generated by the Brownian motion $W$.

Let us consider a market-traded European option with  payoff $h(S_T)$ on expiration date $T$ written on the underlying asset $S$. If the Sharpe ratio $\lambda(t,s) := \frac{\mu(t,s) -r}{\sigma(t,s)}$ satisfies the Novikov condition: $\E\{ \exp (\int_{0}^T\frac{1}{2}{\lambda^2(u,S_u)}\,du ) \}<\infty$, the density process
\begin{align}\label{dqdp}
\frac{d\Q}{d\P}\bigg|_{\F_t} = \exp\left({-\frac{1}{2}\int_{0}^t{\lambda^2(u,S_u)}\,du  + \int_{0}^t{\lambda(u,S_u)}\,dW_u }\right), \qquad 0\le t\le T,
\end{align}
is a $(\P,\mathbb{F})-$martingale  (see \cite{KaratzasShreve91}, Prop. 3.5.12).  This defines a unique equivalent martingale (risk-neutral) measure  $\Q$, and the market price of the option is given by
\begin{align}
V(t,s) = \widetilde{\E}_{t,s}\left\{e^{-r(T-t)} h(S_T)\right\}, \qquad (t,s) \in [0,T]\times \R^+.\label{mkt_pr}
\end{align}
The shorthand notation $\widetilde{\E}_{t,s}\{\cdot\} \equiv \widetilde{\E}\{\cdot|S_t = s\}$  denotes the conditional expectation under  $\Q$. Note that the market price function  $V(t,s)$ does not depend on the drift function $\mu (t,s)$.

Observing the stock and option price movements over time, the  investor  has the timing  flexibility to  sell the option before expiry. While seeking to   maximize the expected discounted market value of the option,  we  incorporate a risk penalty that accounts for the downside risk up to the liquidation time.  Specifically, we  define  the shortfall at time $t$   by
\begin{equation} \label{shortfall}
\ell(t,S_t)=(m-V(t,S_t))^+,
\end{equation}
where $m>0$ is a constant    benchmark  set by the investor.  Then, the  risk penalty is modeled as a \emph{loss function} of the shortfall, denoted by $\psi(\ell(t,S_t))$. Here, the loss function $\psi:\R^+ \to \R$ is assumed to be increasing, convex, continuously differentiable, with $\psi(0)=0$ (see e.g. \citep[Chap 4.9]{FScB}). As a result, the investor faces the penalized optimal stopping problem
\begin{equation} \label{V_prob}\\
J^{\alpha}(t,s) = \sup_{\tau \in \mathcal{T}_{t,T}}\E_{t,s}\left\{e^{-r(\tau-t)}V(\tau,S_\tau)-\alpha \int_t^{\tau} e^{-r(u-t)} \psi\left((m-V(t,S_t))^+\right)du\right\},
\end{equation}
where $\alpha\ge 0$ is a penalization coefficient and $\mathcal{T}_{t,T}$ is the set of $\mathbb F$-stopping times taking values in $[t,T]$.  

Unless otherwise noted, our analysis applies to a general loss function  $\psi$ satisfying the conditions above. Here, let us  give an example to visualize the penalization mechanism. For instance, one can set the benchmark to be the initial  option price, and take $\psi(\ell) = \ell$. Then, the penalty term amounts to accumulating the (discounted) area when the option is below its initial cost. We illustrate this in Figure \ref{realized_shortfall}. Notice that the realized shortfall  stays flat when the option price is above the benchmark, and continues to increase as long as the option is under water.  Other viable specifications include the power penalty $\psi(\ell)=\ell^p$, $p\ge 1$, and the exponential penalty $\psi(\ell)=\exp(\gamma \ell)-1$, $\gamma >0$, and more.

\begin{figure}[!ht]
\centering
\includegraphics[scale=0.65]{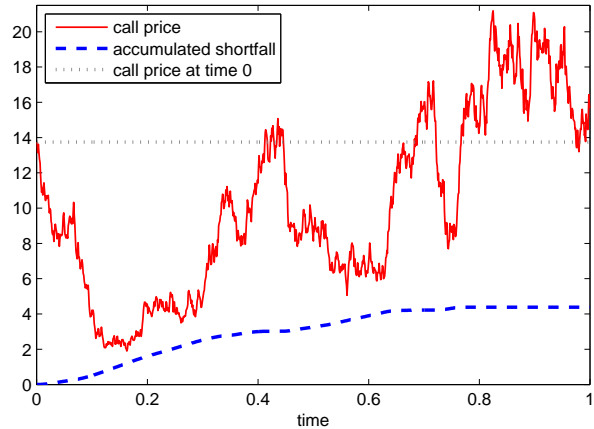}
\caption{\small{The realized shortfall (dashed) based on a simulated price path (solid) of a European call option under the GBM model, with parameters  $S_0=100$, $r=0.03$, $\mu=-0.05$ and $\sigma=0.3$,   $K=100$,   $T=1$, $\alpha =1$. The benchmark $m$ is the initial  call option price.  }}
\label{realized_shortfall}
\end{figure}

In order to quantify the value of optimal waiting, we define the optimal liquidation premium by the difference between the value function $J^{\alpha}$ and the current market price of the option, namely,
\begin{equation}
L^{\alpha}(t,s):= J^{\alpha}(t,s)-V(t,s). \label{delayed_liquidation}
\end{equation}
Alternatively, the optimal liquidation premium $L^{\alpha}$ can be interpreted as the risk-adjusted expected return from a simple buy-now-sell-later strategy.

Denote the discounted penalized liquidation value process by
\begin{align*}
Y_u=e^{-ru}V(u,S_u)-\alpha \int_0^{u} e^{-rt} \psi((m-V(t,S_t))^+)dt.
\end{align*}
In  order to guarantee the existence of an optimal stopping time to problem \eqref{V_prob}, we require that $\E\{\sup_{0\leq u \leq T}Y_u\}<\infty$. For a European call option, the option value $V(t,S_t)$ is dominated by the stock price $S_t$, while the put option price is bounded by the strike price. Consequently, for any linear combination of calls and puts, it suffices to impose $\E\{\sup_{0\leq u \leq T}S_u\}<\infty$. We also require that $\P\{\min_{0\le t\le \hat{t}} S_t>0\}= 1$, which means that the asset price stays strictly positive before  any finite time $\hat{t}$ a.s.  Then by standard optimal stopping theory \citep[Theorem D.12]{KaratzasShreve01}, the optimal liquidation time, associated with $L(t,s)$, is given by
\begin{equation}
\tau^* = \inf\{ \,u \in [t,T]\,:\, L^{\alpha}(u,S_u) =0\, \}.\label{Ltau}
\end{equation}
In other words, it is optimal for the investor to sell the option as soon as the optimal liquidation premium $L^{\alpha}$ vanishes, meaning that the timing flexibility has no value. Accordingly, the investor's optimal liquidation strategy can be described by the sell region $\mathcal S$ and delay region $\mathcal D$, namely,
\begin{align}
\mathcal S&=\{(t, s)\in [0,T]\times \R^+ : \ L^{\alpha}(t,s)=0\},\\
\mathcal D  &=\{(t, s)\in [0,T]\times \R^+ : \ L^{\alpha}(t,s)>0\}.
\end{align}

Our framework can be readily applied to the \emph{reverse} problem of optimally timing to buy an option. This amounts to changing the $\sup$ to $\inf$ in $L^\alpha$. In this paper, we shall focus on the liquidation problem.

\subsection{Analysis of the Optimal Liquidation Premium}\label{sect-ana}
\begin{theorem}\label{prop1_L_G}
Given the underlying price dynamics in \eqref{SDE_S}, the optimal liquidation premium admits the probabilistic representation
\begin{align}
L^{\alpha}(t,s)
&=\sup_{\tau \in \mathcal{T}_{t,T}} \E_{t,s} \left\{ \int_t^\tau e^{-r(u-t)} G^{\alpha}(u,S_u) \,du\right\},\label{L_G}
\end{align}
where  we denote
\begin{align}
G^{\alpha}(t,s):=\big( \mu(t, s) - r\big)sV_s(t,s)-\alpha \psi\left((m-V(t,s))^+\right). \label{Gdrift}
\end{align}
\end{theorem}
\begin{proof} Applying Ito's formula to the market price in \eqref{mkt_pr}, we get
\[\E_{t,s} \left\{e^{-r(\tau-t)}V(\tau,S_\tau)   \right\} -V(t,s) = \E_{t,s}  \left\{ \int_t^\tau e^{-r(u-t)} \big( \mu(u, S_u)  - r\big)S_uV_s(u,S_u) du\right\}.\]
Substituting this into the optimal liquidation premium in \eqref{delayed_liquidation} gives
\begin{equation*}
L^{\alpha}(t,s) = \sup_{\tau \in \mathcal{T}_{t,T}}~\E_{t,s} \bigg\{ \int_t^\tau e^{-r(u-t)}
					\underbrace{\big[\big(\mu(u, S_u)-r\big)S_uV_s(u,S_u)-\alpha \psi((m-V(u,S_u))^+)\big]}_
					{=\,G^\alpha(u,S_u)\text{ in \eqref{Gdrift}}}\,du\bigg\}.
\end{equation*} \end{proof}
We shall call $G^{\alpha}(t,s)$ in  \eqref{Gdrift}  the \emph{drive function}. We observe that it depends on the Delta $V_s \equiv \frac{\partial V}{\partial s}$ of the option and the penalty coefficient $\alpha$ reduces the drive function for every $(t,s)$. Many properties of the optimal liquidation premium $L^\alpha$ can be deduced by studying the drive function.\\

\begin{proposition} \label{trivial_timing}Let $t\in[0,T]$ be the current time.
If the drive function $G^{\alpha}(u,s)$ is positive, $\forall (u,s) \in [t,T]\times \R^+$, then it is optimal to sell at maturity, namely, $\tau^*=T$. ~
If the drive function $G^{\alpha}(u,s)$ is negative, $\forall (u,s) \in [t,T]\times \R^+$, then it is optimal to sell immediately, namely, $\tau^*=t$.
\end{proposition}
\begin{proof}
We observe from the integral in \eqref{L_G} that if the drive function $G^{\alpha}$ is positive (resp. negative), $\forall (u,s) \in [t,T]\times \R^+$, then we can maximize the expectation by selecting the largest (resp. smallest) stopping time, namely, $\tau^*=T$ (resp. $\tau^*=t$).
\end{proof}

In particular, if $V_s(t,s)$ and $(\mu(t,s)-r)$ are of different signs $\forall\,(t,s)$, then the drive function $G^{\alpha}$ is always negative, so it is optimal to sell immediately. Proposition \ref{trivial_timing} can also be applied to the perpetual case if we set  $T= \infty$. In general, the delay  region always contains the region where the drive function is positive, namely,  
\begin{equation} \label{GsubL}
\{G^\alpha>0\}\subset\{L^\alpha>0\},
\end{equation}see e.g. \citep[Prop. 2.3]{OksendalSulemBook}. Intuitively, this means that if $G(t, s) > 0$,  then the investor 
should not sell immediately   since an incremental  positive infinitesimal premium can be obtained   by waiting for an infinitesimally small amount of time.

In addition, we can  infer from \eqref{L_G} the ordering of optimal liquidation premium  based on the drive function.

\begin{corollary}\label{Compare_G} Consider two options $A$ and $B$, along with two penalty coefficients $\alpha_A$ and $\alpha_B$ respectively. If the drive function of $A$ dominates that of $B$, i.e. $G^{\alpha_A}_A(t,s)\geq G^{\alpha_B}_B(t,s),\forall (t,s)\in [0,T]\times \R^+$, then the optimal liquidation premium for $A$, $L^{\alpha_A}_A$, dominates that for $B$, $L^{\alpha_B}_B(t,s)$, i.e. $L^{\alpha_A}_A(t,s) \geq L^{\alpha_B}_B(t,s),\forall (t,s)\in[0,T]\times \R^+$.
\end{corollary}

The corollary allows us to compare the liquidation timing of different penalties. For example, for $0\leq \alpha_1\leq \alpha_2$, we have $G^{\alpha_1}(t,s)\geq G^{\alpha_2}(t,s)$ for the same option. It follows from \eqref{Ltau} and Corollary \ref{Compare_G} that the optimal liquidation time with penalty $\alpha_1$ is later than that with penalty $\alpha_2$.

In general, a variety of delay and sell regions can occur depending on the underlying dynamics and option payoff. Next, we give sufficient conditions so that the delay region is bounded.

\begin{theorem}\label{compactsupport}
Let $T<\infty$ and $S$ be time homogeneous. Then, the delay region is bounded provided that\\
(i) $\exists \,c>0$ s.t. $G^\alpha(t,s)<c$ for every $(t,s)\in[0,T]\times \R^+$; and\\
(ii) there exist constants $b,k>0$ such that $G^\alpha(t,s)<-b$ in $[0,T]\times[k,\infty)$.
\end{theorem}

\begin{proof}
\textbf{Step 1.} \textit{We find a function $\widehat{L}(t,s)$ that  dominates   $L^\alpha(t,s)$ and  is decreasing in both $t$ and $s$.} To this end, we define
\begin{align*}
\widehat{G}^\alpha(s)
	& := \max\{G^\alpha(t,\xi):(t,\xi)\in[0,T]\times[s,\infty)\},\\
\widehat{L}(t,s)
	& := \sup_{\tau \in \mathcal{T}_{t,T}} \E_{t,s}\left\{ \int_t^\tau e^{-r(u-t)}\widehat{G}^\alpha(S_u)\,du,\right\}.
\end{align*}
By construction $\widehat{G}^\alpha:[0,T]\times\R^+ \to \R$ is constant in $t$ and decreasing in $s$. It also satisfies conditions (i) and (ii). Consequently, using the time homogeneity of $S$, we have, for $t>t'$,
\begin{equation*}
\widehat{L}(t,s)
=  		\sup_{\tau \in \mathcal{T}_{0,T-t}} \E_{0,s} \left\{ \int_0^\tau e^{-ru} \widehat{G}^\alpha(S_u) \,du\right\}
\leq 	\sup_{\tau \in \mathcal{T}_{0,T-t'}} \E_{0,s}\left\{ \int_0^\tau e^{-ru} \widehat{G}^\alpha(S_u) \,du\right\}
=		\widehat{L}(t',s).
\end{equation*}
Hence,  $\widehat{L}(t,s)$ is decreasing in $t$.  Moreover, since $\widehat{G}^\alpha$ is decreasing in $s$, we have, for $s'>s$,
\begin{align*}
\widehat{L}(t,s')
& =  \sup_{\tau \in \mathcal{T}_{t,T}} \E_{t,s'} \left\{ \int_t^\tau e^{-r(u-t)} \widehat{G}^\alpha(S_u) \,du\right\}\\
& =  \sup_{\tau \in \mathcal{T}_{t,T}} \E_{t,s} \left\{ \int_t^\tau e^{-r(u-t)} \widehat{G}^\alpha(S_u+s'-s) \,du\right\}\\
& \leq  \sup_{\tau \in \mathcal{T}_{t,T}} \E_{t,s} \left\{ \int_t^\tau e^{-r(u-t)} \widehat{G}^\alpha(S_u) \,du\right\} = \widehat{L}(t,s).
\end{align*}
Therefore,  $\widehat{L}(t,s)$ is also decreasing in $s$.

Since by definition $\widehat{G}^\alpha$ dominates $G^\alpha$,  Corollary \ref{Compare_G} implies  that   $L^\alpha(t,s)$ has a bounded support as long as  $\widehat{L}(t,s)$ has a bounded support.  Henceforth,  we can  assume without loss of generality that $L^\alpha$ is decreasing in both variables $t$ and $s$ and that $G^\alpha$ is time homogeneous and decreasing in $s$. In particular, we denote  $G^\alpha(s) \equiv G^\alpha(t,s)$.

\textbf{Step 2.} \textit{We prove that for every $t>0$, there exists $\hat{s}<\infty$ such that $L^\alpha(t,s)=0$ for every $s>\hat{s}$}. Since  $L^\alpha(t,s)$ is decreasing, it is equivalent  to show that there exists  no $\hat{t}\in (0,T]$ s.t.  $L^\alpha(t,s)>0$ for $0\leq t \leq \hat{t}$ and $s\in \R^+$. To this end, let's suppose that such a time $\hat{t}$ exists.  In other words, $\tau^* = \inf\{t\le u\le T\,:\,  L(u,S_u) =0\} > \hat{t}.$

Now, we show that this leads to a contradiction. Fix $t\in[0,\hat{t})$. Condition (ii) means that there exists $k$ s.t.  $G^\alpha(s)<-b<0$ in $[k,\infty)$. For $s>k$, we let $\tau_k:=\inf\{u \ge t :S_{u}\le k\}$. Since $S$ has continuous paths, we have $\tau_k>t$. Define
\begin{align}
\mathcal{K}(t,s)
:= 	\frac{c}{r}\E_{t,s}\left\{e^{-r(\tau_k-t)}\indic{\tau_k\leq \tau^*}\right\}
  	-\E_{t,s} \left\{b\int_t^{\tau^*\wedge\tau_k}e^{-r(u-t)} \,du\right\} \notag \\
=	\frac{c}{r}\E_{t,s}\left\{e^{-r(\tau_k-t)} \indic{\tau_k\leq \tau^*}\right\}
	-b\left(1-\E_{t,s}\left\{e^{-r(\tau^*\wedge\tau_k-t)}\right\}\right), \label{proofK}
\end{align}where $c$ is the upper bound of $G^\alpha$ in condition (i). 
Next, taking  $s\uparrow\infty$ yields that  $\P_{t,s}(\tau_k\le T) \downarrow 0$, while $\E_{t,s}\left\{e^{-r(\tau^*\wedge\tau_k-t)}\right\}< e^{-r(\hat{t} -t)}$ since $\tau^*>\hat{t}>t$ a.s. Therefore,  we obtain
\begin{equation}\notag
\beta(t,s) :=\frac{c\,\P_{t,s}(\tau_k\leq \tau^*)}{r(1-\E_{t,s}\left\{e^{-r(\tau^*\wedge\tau_k-t)}\right\})}\to 0.
\end{equation}
As a result, for a sufficiently large $s>k$, we get $b \ge \beta(t,s)$, which implies that $\mathcal{K}(t,s)\leq 0$ (see \eqref{proofK}). 

Next we consider the difference
\begin{align}
L^\alpha(t,s)-\mathcal{K}(t,s)
& \leq 	\E_{t,s} \left\{\int_{\tau^*\wedge\tau_k}^{\tau^*} e^{-r(u-t)}G^\alpha(S_u) \,du
		- \frac{c}{r}e^{-r(\tau_k-t)}\indic{\tau_k\leq \tau^*}\right\}\notag \\
& \leq e^{rt}\E_{t,s}\left\{\frac{c}{r}(e^{-r\tau^*\wedge\tau_k}
	 	- e^{-r\tau^*})- \frac{c}{r}e^{-r\tau_k}\indic{\tau^*\leq \tau_k}\right\}\notag \\
& =		-\frac{ce^{rt}}{r}\E_{t,s}\left\{e^{-r\tau^*}\indic{\tau_k\leq \tau^*}\right\}\leq0.\notag
\end{align}
This means that   $L^\alpha(t,s)\leq \mathcal{K} (t,s) \leq 0$. This contradicts the assumption $L^\alpha(t,s)>0$.

\textbf{Step 3.} It remains  to show at time $0$  that  $\exists\,\hat{s}>0$ such that $L^\alpha(0,s)=0$ for every $s>\hat{s}$. Let $\hat{t}\in[0,T]$ and consider for every $t\in[0,T+\hat{t}]$ the optimal stopping problem
\begin{equation*}
\overline{L}^\alpha(t,s) := \sup_{\tau \in \mathcal{T}_{t,T+\hat{t}}} \E_{t,s} \left\{ \int_t^\tau e^{-r(u-t)} G^\alpha(S_u) \,du\right\}.
\end{equation*}
The time homogeneity of $S$ yields that $\overline{L}^\alpha(\hat{t},s)=L^\alpha(0,s)$. Now  we apply  Step 2 and conclude  that  there exists $\hat{s}>0$ such that $\overline{L}^\alpha(\hat{t},s)=0$ for every $s>\hat{s}$. Hence,  the delay region is bounded above.
\end{proof}

We remark that the statement and the proof of Theorem \ref{compactsupport}   do not involve the properties of the loss function. In other words, as long as the resulting drive function satisfies conditions (i) and (ii), the delay region is bounded.  We notice that if the delay region is bounded, then there exists a constant $\bar{s}$ such that $\{L^\alpha>0\}\subseteq[0,T]\times (0,\bar{s})$.  We will utilize  Theorem \ref{compactsupport} repeatedly when we discuss the liquidation strategies in subsequent sections.

\subsection{Applications to GBM and Exponential OU Underlyings} \label{sect:GBM_OU}
Henceforth, we shall investigate analytically and numerically the optimal liquidation timing when $S$ follows (i) the  geometric Brownian motion (GBM) model with  $\mu(t,s)=\mu$ and $\sigma(t,s)=\sigma>0$, as well as (ii) the exponential OU model with  $\mu(t,s)=\beta(\theta-\log(s))$ and $\sigma(t,s)=\sigma>0$.

We will study the liquidation timing of a stock, European put and call options. For both the GBM and exponential OU cases, the risk-neutral measure $\Q$ is uniquely defined by \eqref{dqdp}, and  the Novikov condition is satisfied (see Appendix \ref{appendix_A}). Furthermore, the $\Q$ dynamics of $S$ is a GBM with drift $r$ and the no-arbitrage prices (see \ref{mkt_pr}) of a call and a put with strike $K$ and maturity $T$ are given by
\begin{equation}
C(t,s)=s\,\Phi(d_1)-Ke^{-r(T-t)}\Phi(d_2), \quad \text{ and } \quad  P(t,s)=Ke^{-r(T-t)}\Phi(-d_2)-s \,\Phi(-d_1), \label{CP}
\end{equation}
where $\Phi$ is the standard normal c.d.f. and
\begin{equation}
d_1=\frac{\log(\frac{s}{K})+(r+\frac{\sigma^2}{2})(T-t)}{\sigma\sqrt{T-t}} , \qquad d_2=d_1-\sigma\sqrt{T-t}.  \notag
\end{equation}

In order to  numerically compute   the  non-trivial liquidation strategy, we   solve the variational inequality (VI) of the form
\begin{align}
\text{min}\bigg\{-L^{\alpha}_t-\mu(t,s)s\,L^{\alpha}_s-\frac{\sigma^2(t,s)s^2}{2}L^{\alpha}_{ss} + r L^{\alpha}-G^{\alpha}, \ L^{\alpha}\bigg\}=0, \label{L_general_VI}
\end{align}
with terminal condition $L^{\alpha}(T,s)=0$ and where $(t,s) \in [0,T) \times \mathbb{R}^+$. In Section  \ref{Sect-ExUn}, we show that the above VI admits a  unique strong solution in the terminology of  \cite{Bensoussan78} under conditions that include the GBM and exponential OU cases (see Theorem \ref{ExistenceStrongSolution}).  For  implementation, we adopt the  Crank-Nicholson scheme for the VI \eqref{L_general_VI} on a finite (discretized) grid  $D=[s_{\min},s_{\max}]\times[0,T]$.  We refer to the book by Glowinski \citep[Chap.3]{Glowinski} for details on numerical methods for solving inhomogeneous VIs of parabolic type.

\section{Optimal Liquidation with a GBM Underlying}\label{sect-GBM}
We begin our first series of illustrative examples under the GBM model. In view of Proposition \ref{trivial_timing}, we observe that, if $\mu\leq r$, it is never optimal to hold a stock or a call (or in general any positive delta position) regardless whether we introduce a risk penalty or not. On the other hand, Proposition \ref{trivial_timing} also implies that, if $\mu>r$ and $\alpha=0$, it is always optimal to delay.

 However,  with a non-zero risk penalty ($\alpha>0$), the solution can be non-trivial. To see this, we note that the drive function associated with  a call is given by $G^\alpha_{Call}(t,s)=(\mu-r)sC_s-\alpha\psi((m-C(t,s))^+)$ where $C(t,s)$ is the call price in  \eqref{CP}. In particular, the penalty term  is strictly positive at $s=0$ and  decreasing, and it vanishes for large $s$. On the other hand, the first term  $(\mu-r)sC_s$ is   strictly increasing from zero at $s=0$. This implies that there exists a price level   $\hat{s}$  such that  $G^\alpha_{Call}(t,s)$ is  positive in $[0,T]\times[\hat{s},\infty)$.  In turn, it follows from  \eqref{GsubL}  that the sell region must be  bounded (possibly empty)  and the delay region is unbounded. The same argument applies to the case with a stock. Figure \ref{Stock_Call_SF_1} illustrates this.

\begin{figure}[h!]
\centering
\includegraphics[scale=0.55]{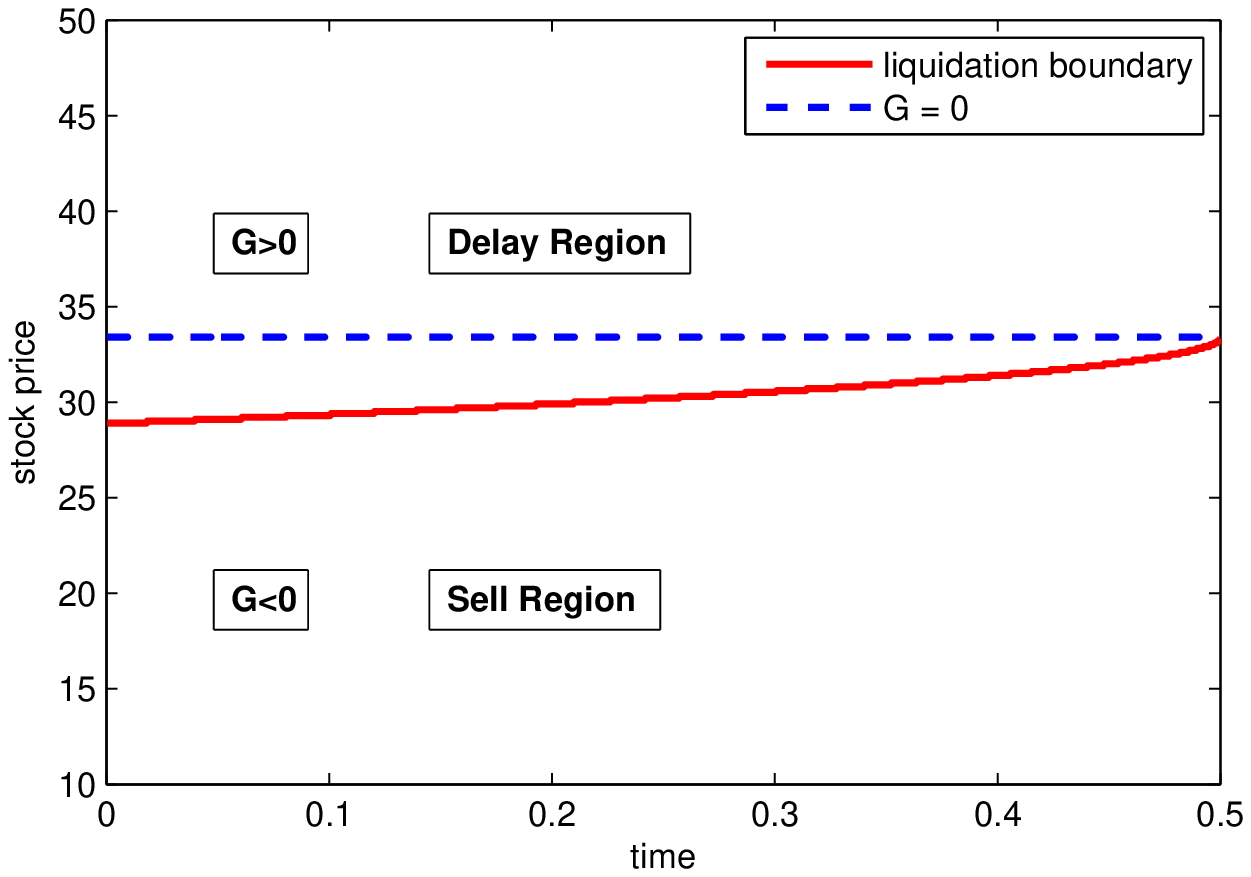}\qquad ~~
\includegraphics[scale=0.55]{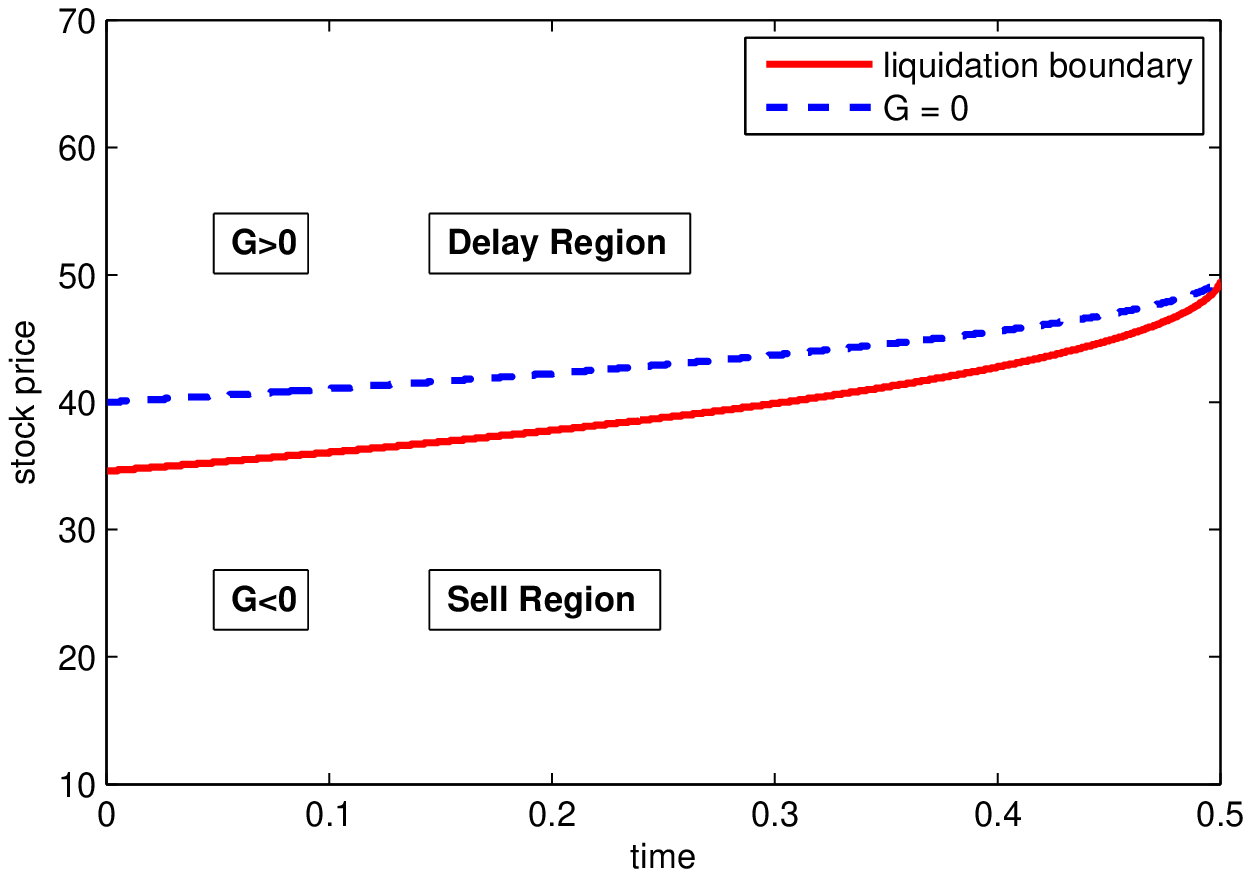}
\caption{\small{The optimal liquidation boundaries (solid) and the zero contours of $G^{\alpha}$ (dashed) of a stock (left panel) and a call option (right panel). We take $T=0.5$, $r=0.03$, $\mu=0.08$, $\sigma=0.3$, $K=50$, $\alpha=0.1$. The loss function is given by $\psi(\ell)=\ell$, with the benchmark $m=50$ for the stock and $m=C(0,K)$ for the call. }}
\label{Stock_Call_SF_1}
\end{figure}

Next, we consider the liquidation of a put option. Recall the put price $P(t,s)$ given in \eqref{CP}. Its  negative Delta implies that for $\mu\geq r$ the drive function $G_{Put}^\alpha(t,s) \le 0$,\, $\forall (t,s)$, meaning that  it is optimal to sell immediately by Proposition \ref{trivial_timing}. In contrast, when $\mu<r$, the sell region is empty if  $\alpha=0$, but under risk penalization  the optimal strategy  may be non-trivial. \\

\begin{proposition}\label{Put4}
Consider the optimal liquidation of a put under the GBM model with $\mu<r$ and $\alpha>0$. Then, the delay region  is  {bounded}. 
Furthermore, it is non-empty if $m<K$  and $\exists \, \hat{t} \in [0,T]$ such that $\alpha\psi'((m-P(\hat{t},0))^+)<r-\mu$.
\end{proposition}
\begin{proof}
The drive function for the put $G^\alpha \equiv G^\alpha_{Put}(t,s)=(r-\mu)s\Phi(-d_1)-\alpha\psi((m-P(t,s))^+)$ satisfies
\begin{align}
\lim_{s\to 0} G^\alpha(t,s)&=-\alpha\psi((m-K)^+) \leq 0, \label{Put1} \\
\lim_{s\to \infty} G^\alpha(t,s)&=-\alpha\psi(m) < 0, \label{Put2} \\
\frac{\partial G^\alpha}{\partial s}(t,s)&=[r-\mu-\alpha\psi'((m-P(t,s))^+)\indic{m>P(t,s)}]\Phi(-d_1)-(r-\mu)s\Gamma, \label{Put3}
\end{align}
where $\Gamma=\frac{\partial^2 P}{\partial s^2}\geq 0$. In turn, we fix any $\hat{b}\in (0,\alpha \psi (m))$ and define $\overline{\psi}(\ell):=\min\{\psi(\ell),\hat{b}\}$. This implies the inequality\begin{equation*}
\overline{G}^\alpha(t,s):=(r-\mu)s\Phi(-d_1) - \alpha \overline{\psi}((m-P(t,s))^+) \geq G^\alpha(t,s).
\end{equation*}
Then by Corollary \ref{Compare_G}, we only need to show that $\overline{G}^\alpha$ satisfies the assumptions of Theorem \ref{compactsupport}. We observe that $\overline{G}^\alpha$ is bounded above and it follows from  \eqref{Put2} that $\lim_{s\to\infty}\overline{G}^\alpha(t,s)\to -\alpha \hat{b} <0$ for every $t\in[0,T]$. Moreover, there exists $\hat{s}>0$ such that for every $s>\hat{s}$, $\overline{\psi}((m-P(t,s))^+)=\hat{b}$. Consequently, we have 
\begin{equation*}
\frac{\partial \overline{G}^\alpha}{\partial t}
= (\mu-r)s\phi(d_1)\frac{\log(\frac{s}{K})-(r+\frac{\sigma^2}{2})(T-t)}{2\sigma(T-t)^\frac{3}{2}}\leq 0,
\end{equation*}
for $s>\max\{\hat{s}, K\exp((r+\sigma^2/2)T)\}$ and $t\in[0,T]$.  Since $\overline{G}^\alpha (0,s)\to -\alpha \hat{b}$ as $s\to-\infty$, we can choose  $b\in(0,\alpha \hat{b})$ such that   $\exists k>\max\{\hat{s}, K\exp((r+\sigma^2/2)T)\}$ and  $-b>\overline{G}(0,s)>\overline{G}(t,s)$ in  $[0,T]\times [k,\infty)$. Therefore,  $\overline{G}$ satisfies the assumptions of Theorem \ref{compactsupport}. 

Finally, suppose $\exists \, \hat{t} \in [0,T]$ such that $\alpha\psi'((m-P(\hat{t},0))^+)<r-\mu$, where $m<K$. It follows from \eqref{Put1} and \eqref{Put3} that  $G^{\alpha}(\hat{t},0)=0$ and $\frac{\partial G^{\alpha}}{\partial s}(\hat{t},0)>0$, so that the set $\{G^\alpha>0\}$ is non-empty. In turn, the inclusion  \eqref{GsubL} implies that  the delay region is also non-empty.
\end{proof}

\begin{remark}  As an example, the delay region is empty if $\alpha\psi'((m-P(t,0))^+)\geq r-\mu>0,\,\forall \,t \in [0,T]$. Indeed, since we have  $G^\alpha(t,0)\leq 0$ and $\frac{\partial G^{\alpha}}{\partial s}(t,0)\leq 0,\,\forall\,t\in[0,T]$, $G^\alpha$ cannot be strictly positive. By Proposition \ref{trivial_timing}, it is optimal to sell immediately.
\end{remark}

Proposition \ref{Put4} is illustrated in Figure \ref{Put_SF}. In these examples, the delay region is non-empty and the sell region is unbounded but may  be disconnected  (Figure \ref{Put_SF} (right)). This can arise when, for example, $G^\alpha(t,0)<0$ for every $t\in[0,T]$, but $\min_t\max G^\alpha(t,s)>0$. The intuition for a disconnected sell region is as follows. If the put is deeply in the money (i.e. when $S_t$ is close to zero), its market price  has very limited room to increase since it is bounded above $Ke^{-r(T-t)}$. At the same time, delaying sale further will incur  a penalty. Therefore, when the penalization coefficient $\alpha$ is high,  it is optimal to sell at a low  stock price level. On the other hand,  if   the put is deep out of the money (i.e. when $S_t$ is very high), the market price and the Delta  of the put are   close to zero, meaning the drive function becomes more negative  and selling immediately is optimal.

\begin{figure}[h!]
\centering
\includegraphics[scale=0.55]{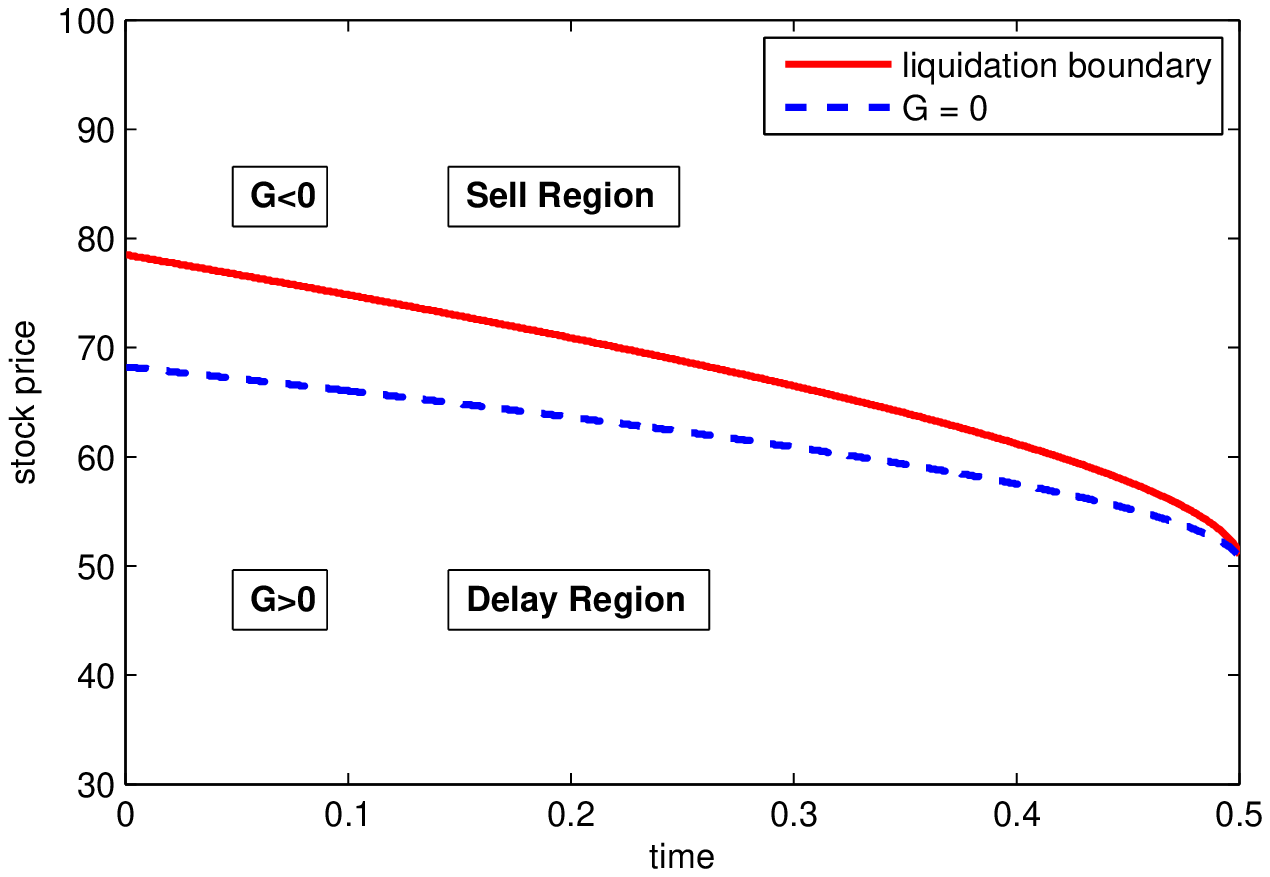}\qquad ~~
\includegraphics[scale=0.55]{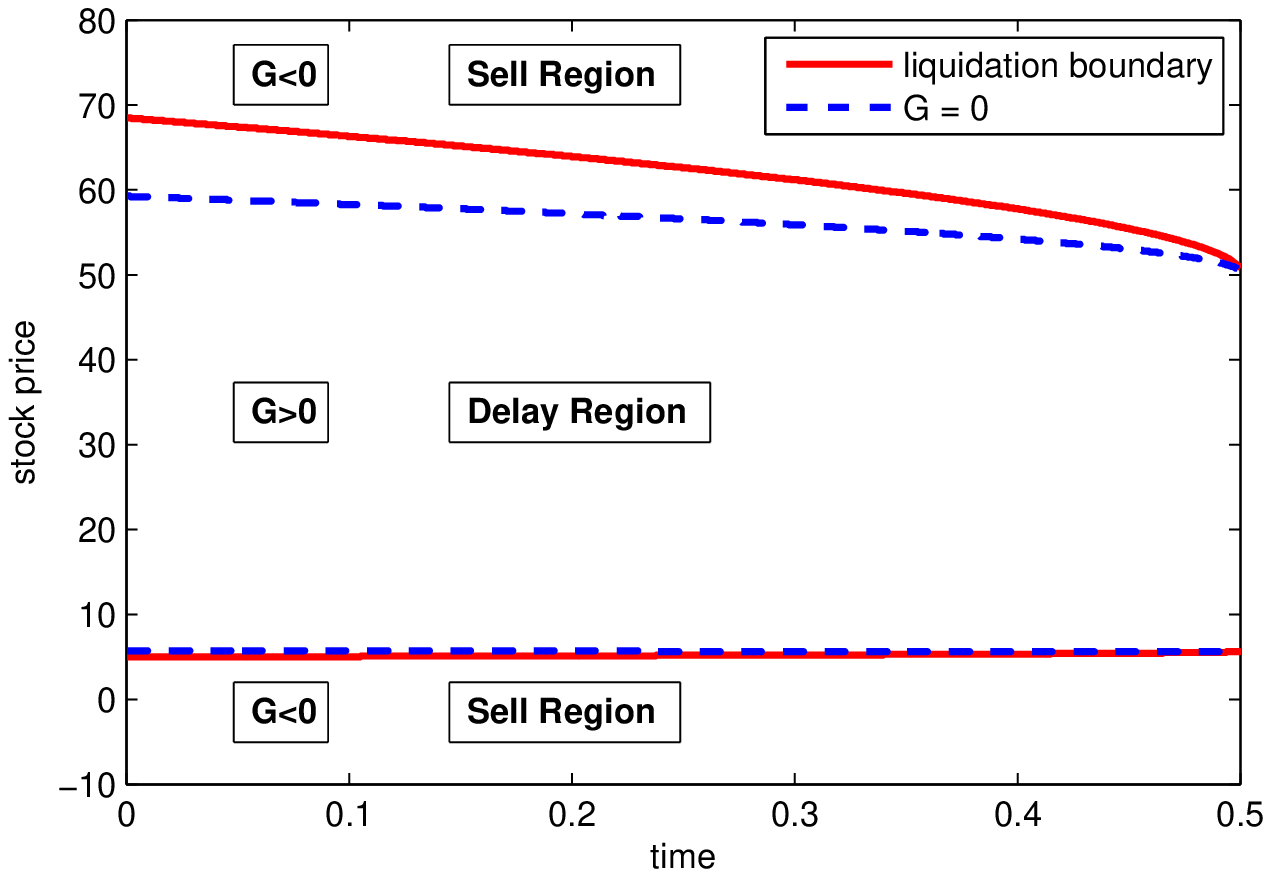}
\caption{\small{The optimal liquidation boundary (solid) and the zero contour of $G^{\alpha}$ (dashed) of a put option under GBM dynamics with the loss function $\psi(\ell)=\ell$.  We take $m=2K$, $\alpha=0.001$ (left panel), and $m=P(0,K)$, $\alpha=0.01$ (right panel). Parameters: $T=0.5$, $r=0.03$, $\mu=0.02$, $\sigma=0.3$, $K=50$.}}
\label{Put_SF}
\end{figure}

For a long position in   calls and the underlying stock, or in puts, the Delta $C_s$ takes a constant sign. As an example of a derivative with a Delta of non-constant sign, we consider a long straddle. This is a combination of a call and a put with strike prices $K_1 \leq K_2$ respectively and the same maturity $T$. The payoff of a straddle is given by $h^{STD}(S_T):=(S_T-K_1)^+ + (K_2-S_T)^+$. The market price of a long straddle, denoted by $C^{STD}$, is simply the sum of the respective Black-Scholes call and put prices, i.e. $C^{STD}(t,s)=C(t,s)+P(t,s)$. For simplicity, we   set $K_1=K_2=K$.
\begin{proposition}\label{STD}
For the optimal  liquidation of a long straddle position under the GBM model, it follows that  \\
(i) if $\mu=r$, the  delay region must be empty;\\
(ii) if $\mu>r$, the delay region is unbounded;\\
(iii) if  $\mu<r$, the delay region is bounded.
\end{proposition}
\begin{proof}
The straddle's drive function is $G^{\alpha}_{STD}(t,s)=(\mu-r)sC_s^{STD}(t,s)-\alpha \psi((m-C^{STD}(t,s))^+)$. For $\mu=r$, the conclusion follows immediately by Proposition \ref{trivial_timing}. If $\mu>r$ we simply notice that $G_{STD}^\alpha(t,s) \to \infty$ as $s\to\infty$ for every $t\in[0,T]$, and the assertion follows from the inclusion \eqref{GsubL}.

Now suppose $\mu<r$. We will show that $G^\alpha$ satisfies the assumptions of Theorem \ref{compactsupport}. Clearly, $G_{STD}^\alpha$ is bounded above. Since $C^{STD}(t,s) \to \infty$ as $s\to\infty$ for every $t\in[0,T]$, then there exists $\hat{s}>0$ such that, for every $s>\hat{s}$ and $t\in[0,T]$, $\psi((m-C^{STD}(t,s))^+) = 0$. Moreover, for $s>\max\{\hat{s},K\exp\left((r+\sigma^2/2)T\right)\}$, we have
\begin{equation*}
\frac{\partial \Phi(d_1)}{\partial t} = \phi(d_1)\frac{\log(\frac{s}{K})-(r+\frac{\sigma^2}{2})(T-t)}{2\sigma(T-t)^{\frac{3}{2}}} > 0,
\end{equation*}
and thus
\begin{equation*}
\frac{\partial G_{STD}^\alpha}{\partial t}(t,s)=2(\mu-r)s\frac{\partial \Phi(d_1)}{\partial t} \leq 0.
\end{equation*}
This implies $G_{STD}^\alpha(0,s)\geq G_{STD}^\alpha(t,s)$ for every $t\in[0,T]$. Since $G_{STD}^\alpha(0,s)\to -\infty$ as $s\to \infty$, for a fixed $b>0$ there exists $k_b>0$ such that $G^\alpha_{STD}(0,s)<-b$ for every $s\geq k_b$. Therefore, setting $k=\max\{\hat{s},K\exp\left((r+\sigma^2/2)T\right),k_b\}$, we have $G^\alpha_{STD}(t,s)\leq G^\alpha_{STD}(0,s)<-b$ in $[0,T]\times[k,\infty)$. Therefore, the assumptions of Theorem \ref{compactsupport} are satisfied and we conclude.
\end{proof}

In particular, Proposition \ref{STD}  suggests that when  $\mu<r$,  the sell region is  unbounded, even if $\alpha=0$. In Figure \ref{Straddle_SF2}, we illustrate the optimal liquidation boundaries for cases (ii) and (iii).  When the investor is bullish (left panel: $\mu = 0.08 > 0.03 = r$), the liquidation boundary is increasing and the delay region is on top of the sell region. Interestingly, the opposite is observed when the investor is bearish (right panel: $\mu = 0.02 < 0.03 = r$).

\begin{figure}[h!]
\centering
\includegraphics[scale=0.55]{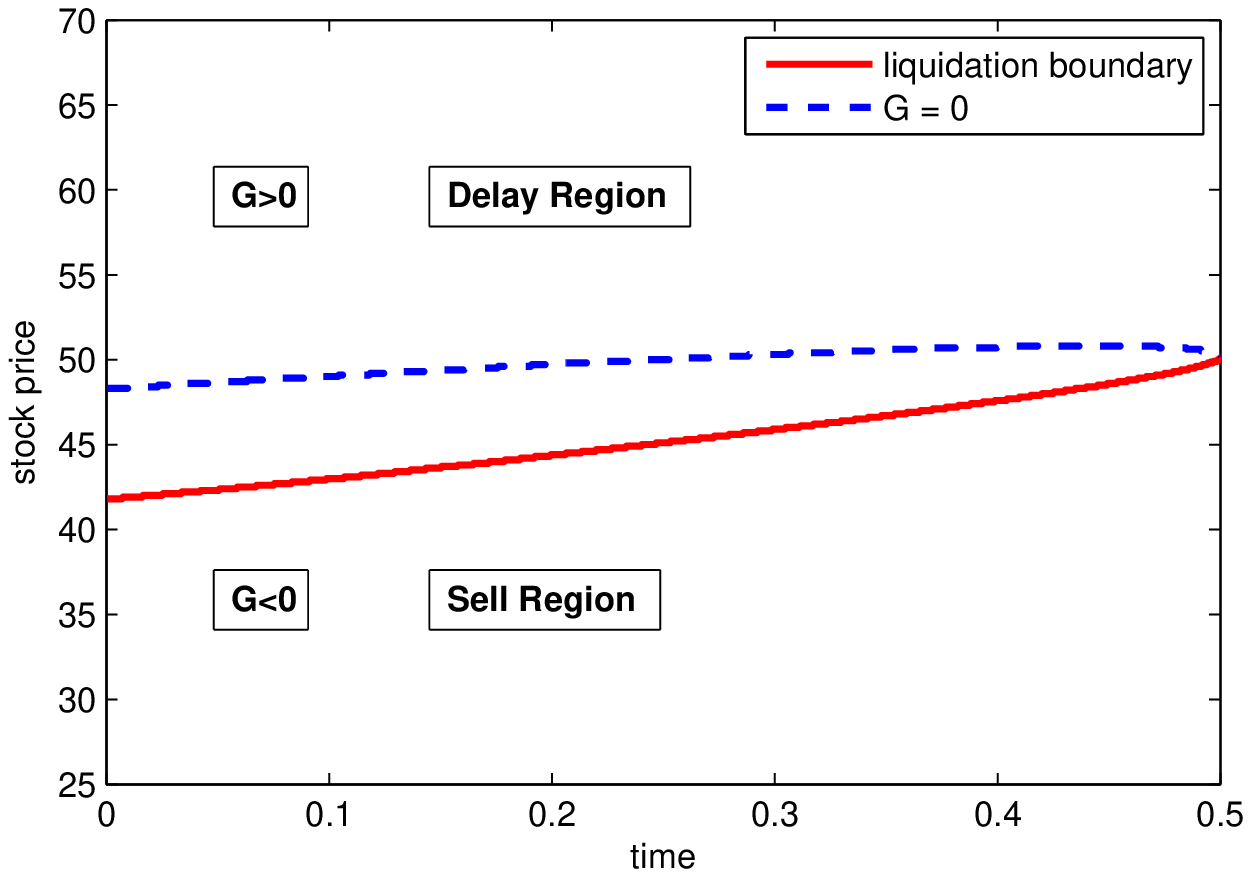}\qquad ~~
\includegraphics[scale=0.55]{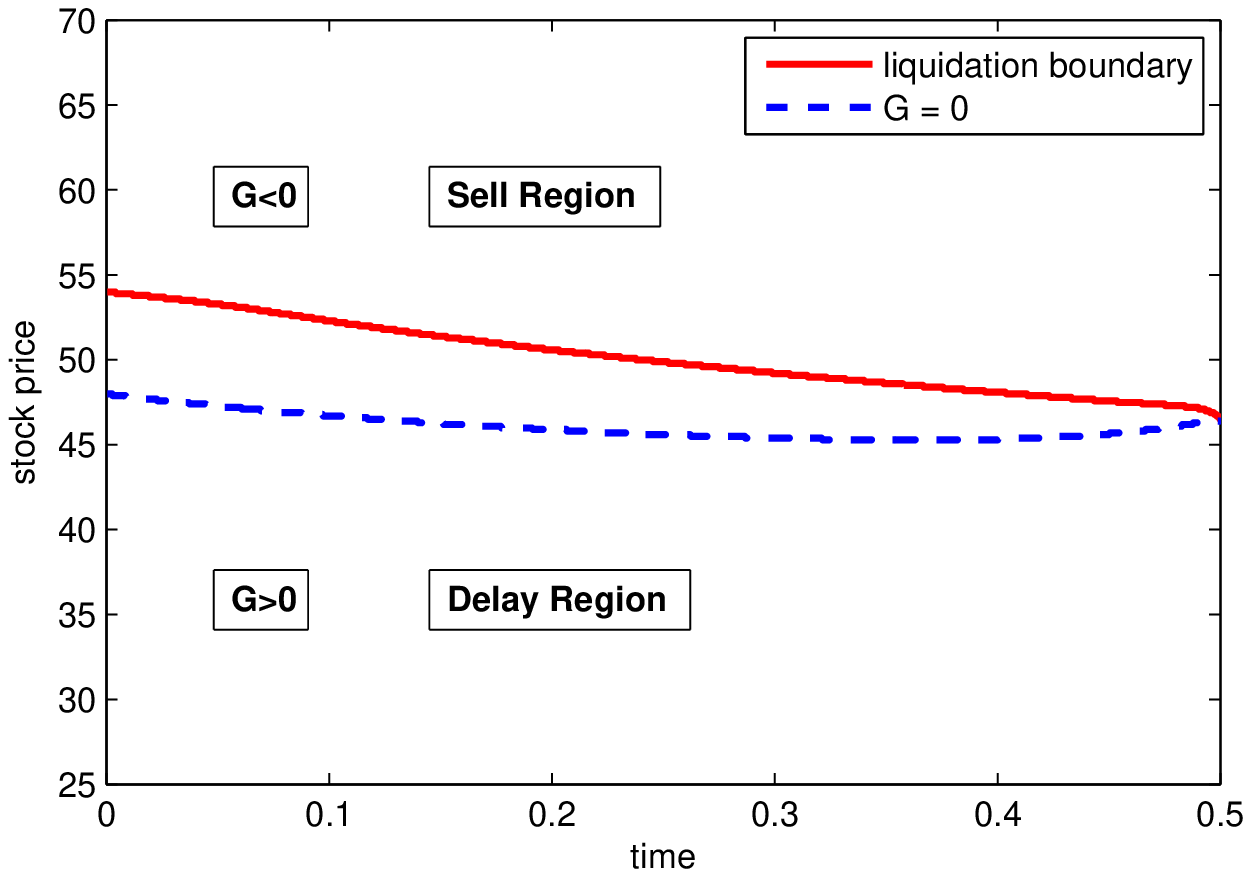}
\caption{\small{Optimal liquidation boundary and the zero contour of $G^{\alpha}$ for a straddle under the  GBM  model  with the loss function $\psi(\ell)=\ell$. We set $K=50$, $m=C^{STD}(0,K)$, $\alpha=0.1$, $r=0.03$, $\mu=0.08$ (left panel) and $\mu=0.02$ (right panel).}}
\label{Straddle_SF2}
\end{figure}

We end this section by discussing the liquidation timing of a stock with an infinite horizon ($T=\infty$). This leads to the following stationary optimal stopping problem 
\begin{equation}\label{PerpetualStock}
L(s) = \sup_{\tau \in \mathcal{T}} \E_{s} \left\{ \int_0^\tau e^{-ru} G^{\alpha}(S_u) \,du\right\}.
\end{equation}
where $G^\alpha(s) = (\mu-r)s - \alpha \psi((m-s)^+)$ and   $\mathcal{T}$ is the set of $\Fil$-stopping times taking values in $[0,\infty]$. When $\mu\leq r$, selling immediately is  optimal according to  Proposition \ref{trivial_timing}, as for the case with finite maturity. As it turns out,  the liquidation problem has the opposite  trivial solution when  $\mu>r$,  that is,  it is optimal to hold forever.  \\

\begin{proposition}
If $\mu>r$, then the value function $L(s)$ in \eqref{PerpetualStock} is infinite and it is optimal to never sell the stock.
\end{proposition}
\begin{proof}
Consider a candidate stopping time $\tau=\infty$. Then, by applying Tonelli's theorem, we have
\begin{align*}\notag
\E_s\left\{\int_0^{\infty} e^{-ru}G^\alpha(S_u)du\right\}
& = \E_s\left\{\int_0^{\infty}e^{-ru}(\mu-r)S_udu - \alpha\int_0^{\infty} e^{-ru}\psi((m-S_u)^+)du\right\} \notag \\
& = \int_0^{\infty}e^{(\mu-r)u}(\mu-r)sdu - \alpha\int_0^{\infty} e^{-ru}\E\left\{\psi((m-S_u)^+)\right\}du \notag \\
& \geq \int_0^{\infty}e^{(\mu-r)u}(\mu-r)sdu - \alpha\int_0^{\infty} e^{-ru}\psi(m)du = \infty, \notag
\end{align*}
since $\mu>r$ and $\psi$ is increasing. Hence, $L(s)=\infty$ and it is never optimal to sell.
\end{proof}

\section{Optimal Liquidation with an Exponential OU Underlying}\label{sect-OU}
In the exponential OU model, the stock price satisfies the SDE
\begin{equation}\label{MRSDE}
dS_t = \beta(\theta - \log S_t )S_t\,dt + \sigma S_t dW_t,
\end{equation}
with $\theta \in\R$ and $\beta, \sigma >0$. Therefore, the optimal liquidation premium $L(t,s)$ is given by equation \eqref{L_G} with the drive function 
\begin{equation}\label{GOUgen}
G^{\alpha}(t,s)=[\beta(\theta-\log(s))-r]sV_s(t,s)-\alpha \psi((m-V(t,s))^+),
\end{equation}where $V(t,s)$ is a generic option price in \eqref{mkt_pr}.

In contrast to the GBM case,  the optimal  liquidation strategy can now be non-trivial for  a  stock or a  call  when there is no penalty. More generally, we can prove that the delay region is in fact  bounded. The intuition should be clear: when  $S_t$ is very high, it is expected to revert back to its long-term mean, so that selling immediately becomes optimal.

\begin{proposition}\label{prob_repr_OU}
Under the exponential OU model, the delay region for a call is bounded.
\end{proposition}
\begin{proof}
 The drive function $G^\alpha_{Call}$ for the call is given by \eqref{GOUgen} with $V(t,s) = C(t,s)$ (see \eqref{CP} for the call price). It is bounded above, so it satisfies condition (i) of Theorem \eqref{compactsupport}. As is well known, the call price satisfies $\frac{\partial C(t,s)}{\partial t}\leq0$.  In addition,  $\beta(\theta-\log(s))-r \leq 0$ iff  $s \geq \exp(\theta-\frac{r}{\beta})$, and     $\frac{\partial \Phi(d_1)}{\partial t} \geq 0$ for $s \geq K\exp\left((r+\sigma^2/2)T\right)$. In turn, we have
\begin{equation*}
\frac{\partial G^\alpha_{Call}}{\partial t}(t,s) = [\beta(\theta-\log(s))-r]s\frac{\partial \Phi(d_1)}{\partial t}
											+ \alpha\frac{\partial C(t,s)}{\partial t}\psi'((m-C(t,s))^+)\indic{m>C(t,s)} \leq 0, 
\end{equation*} for $s>\max\{\exp(\theta-\frac{r}{\beta}), K\exp\left((r+\sigma^2/2)T\right)\}$ and  $t\in[0,T]$.  This  implies $G^\alpha_{Call}(0,s) \geq G^\alpha_{Call}(t,s)$. Fix $b>0$. Since $G^\alpha_{Call}(0,s)\to -\infty$, $\exists k_b>0$ s.t., $\forall s>k_b$, $G^\alpha_{Call}(0,s)<-b$. Hence, if we set $k=\max\{\exp(\theta-\frac{r}{\beta}), K\exp\left((r+\sigma^2/2)T\right), k_b\}$, we are guaranteed that $G^\alpha_{Call}(t,s)\leq G^\alpha_{Call}(0,s)<-b$ in $[0,T]\times[k,\infty)$, thus satisfying condition (ii) of Theorem  \ref{compactsupport}. As a result,  Theorem \ref{compactsupport} applies and gives the  boundedness of the delay region for a call.
\end{proof}
Since    a stock can be viewed as   a call with strike $K=0$, Proposition \ref{Call_OU} also applies to the optimal liquidation of a stock   over a  finite time horizon.  Also, we notice  the delay region   can be empty, and we can   identify  this case by finding the maximum of the drive function. As an example, we consider the case of the stock with penalty function $\psi((m-S_t)^+)=(m-S_t)^+$, and we obtain the maximizer of $G^\alpha$ in different scenarios
\begin{equation*}
\arg\max G^\alpha = \begin{cases}
			\exp(\theta-1-\frac{r-\alpha}{\beta}) & \text { if } \exp(\theta-1-\frac{r-\alpha}{\beta}) < m,\\
			\exp(\theta-1-\frac{r}{\beta}) & \text { if } \exp(\theta-1-\frac{r}{\beta}) > m,\\
			m & \text { otherwise},
			\end{cases}
\end{equation*}
and the corresponding maximum values
\begin{equation*}
\max G^\alpha = \begin{cases}
			(\beta-\alpha)\hat{s}_1-\alpha(m-s^*_1) & \text { if } \hat{s}_1 < m,\\
			\beta \hat{s}_2 & \text { if } \hat{s}_2 > m,\\
			(\beta(\theta-\log(m))-r)m & \text { otherwise,}
			\end{cases}
\end{equation*}
where
\begin{equation*}
\hat{s}_1 = \exp(\theta-1-\frac{r-\alpha}{\beta}), \qquad \hat{s}_2 = \exp(\theta-1-\frac{r}{\beta}).
\end{equation*}
Thus, the delay region is non-empty if and only if $\max G^\alpha>0$.

The optimal liquidation boundary for stock is shown in Figure \ref{Stock_OU} for $\alpha=0$ (left panel) and $\alpha>0$ (right panel). We notice that, in both cases, the optimal strategy is to sell immediately if $S_t$ is high enough. Intuitively, if $S_t$ is high, it is expected to revert back to its long-term mean, so selling immediately becomes optimal. However, if $S_t$ is low, the optimal behavior depends on the parameter $\alpha$. On one hand, $S_t$ is expected to increase and thus the investor should wait to sell at a better price (Figure \ref{Stock_OU}, right panel). On the other hand, such benefit is countered (if the penalization coefficient is high enough) by the risk incurred from holding the position, and this induces the investor to sell immediately. As a consequence, the sell region is disconnected (Figure \ref{Stock_OU}, right panel).
\begin{figure}[ht]
\centering
\includegraphics[scale=0.55]{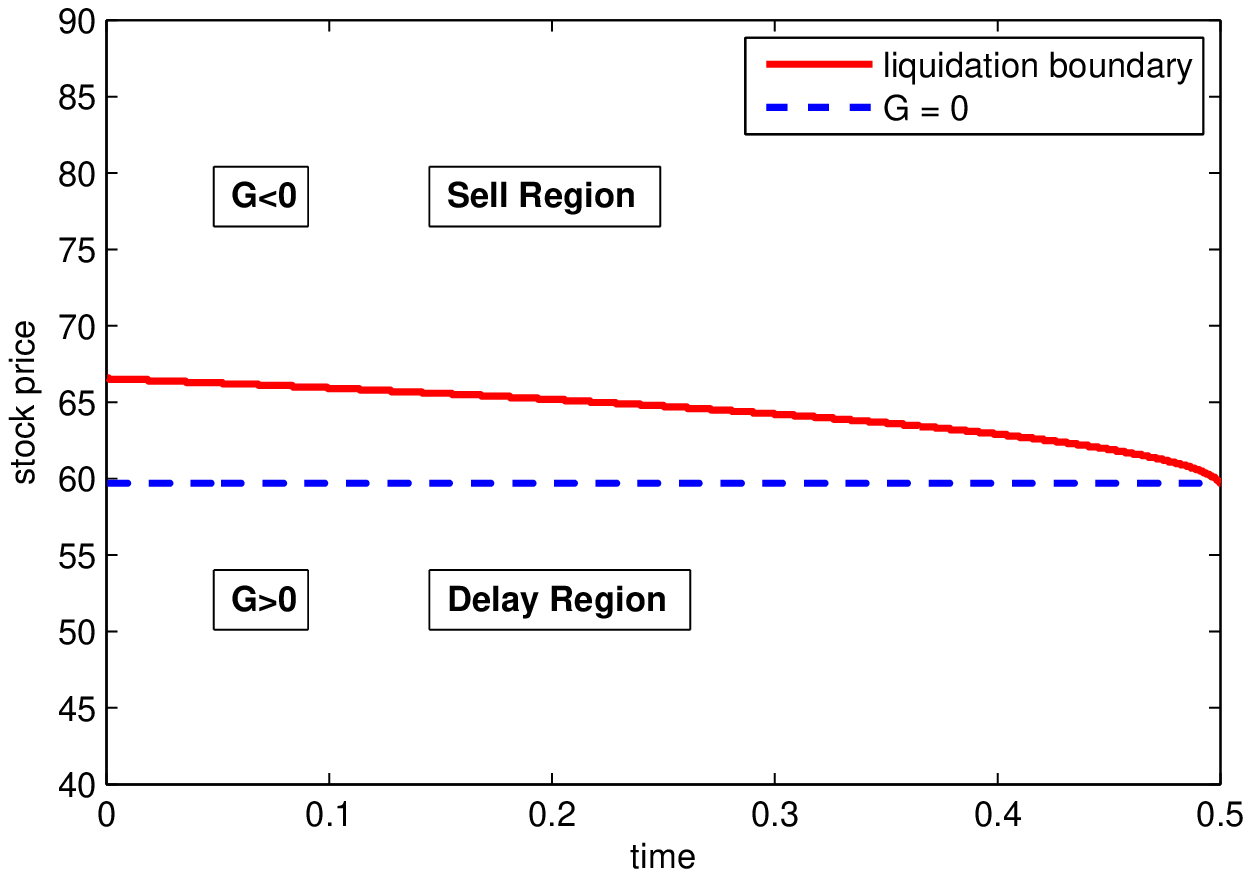}\qquad ~~
\includegraphics[scale=0.55]{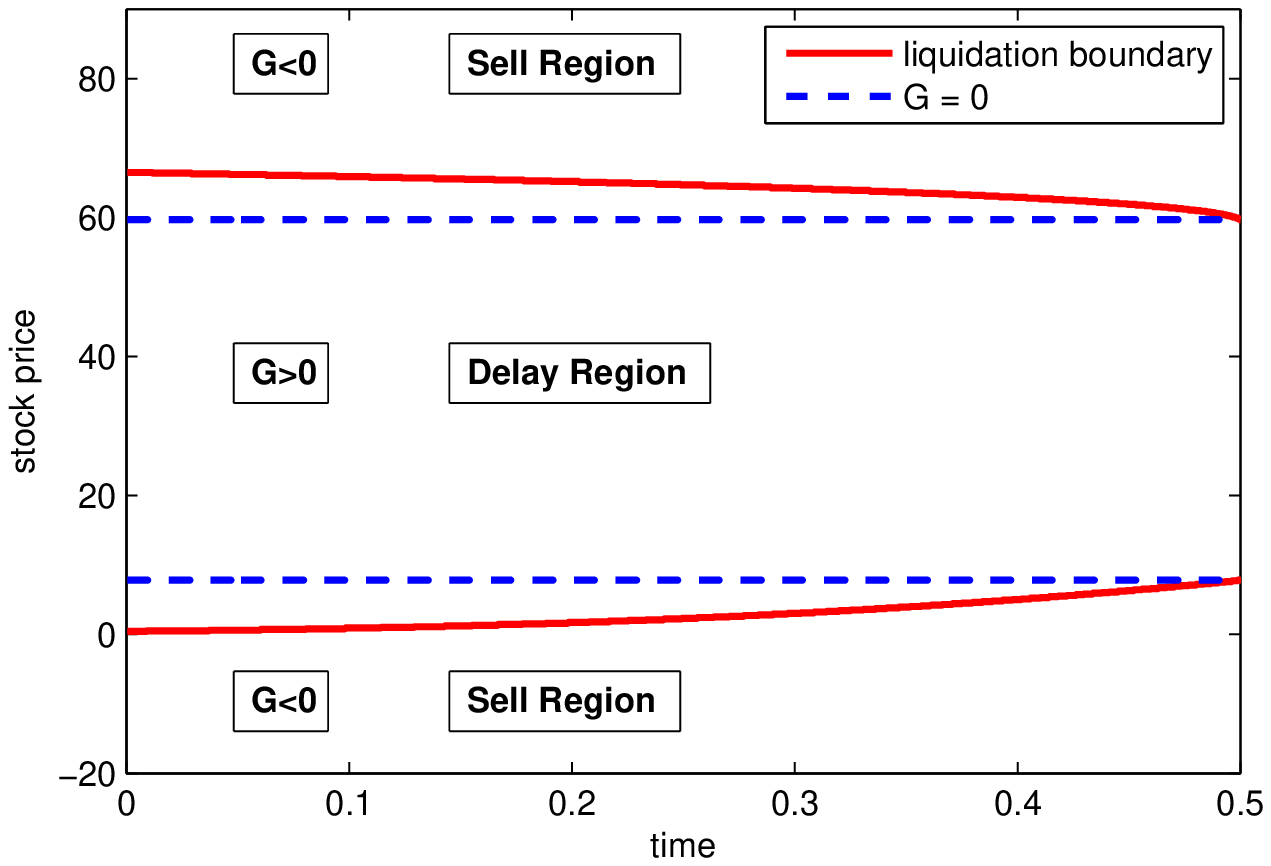}
\caption{\small{The liquidation boundary (solid) and the zero contour of $G^{\alpha}$ (dashed) for a stock under exponential OU dynamics. Parameters: $T=0.5$, $r=0.03$, $\theta=\log(60)$, $\beta=4$, $\sigma=0.3$, $\psi(\ell)=\ell$, $\alpha=0$ (left), $\alpha=1.5$ (right). }}
\label{Stock_OU}
\end{figure}

Figure \ref{Call_OU} illustrates the delay region for a call option with penalty.
\begin{figure}[ht]
\centering
\includegraphics[scale=0.55]{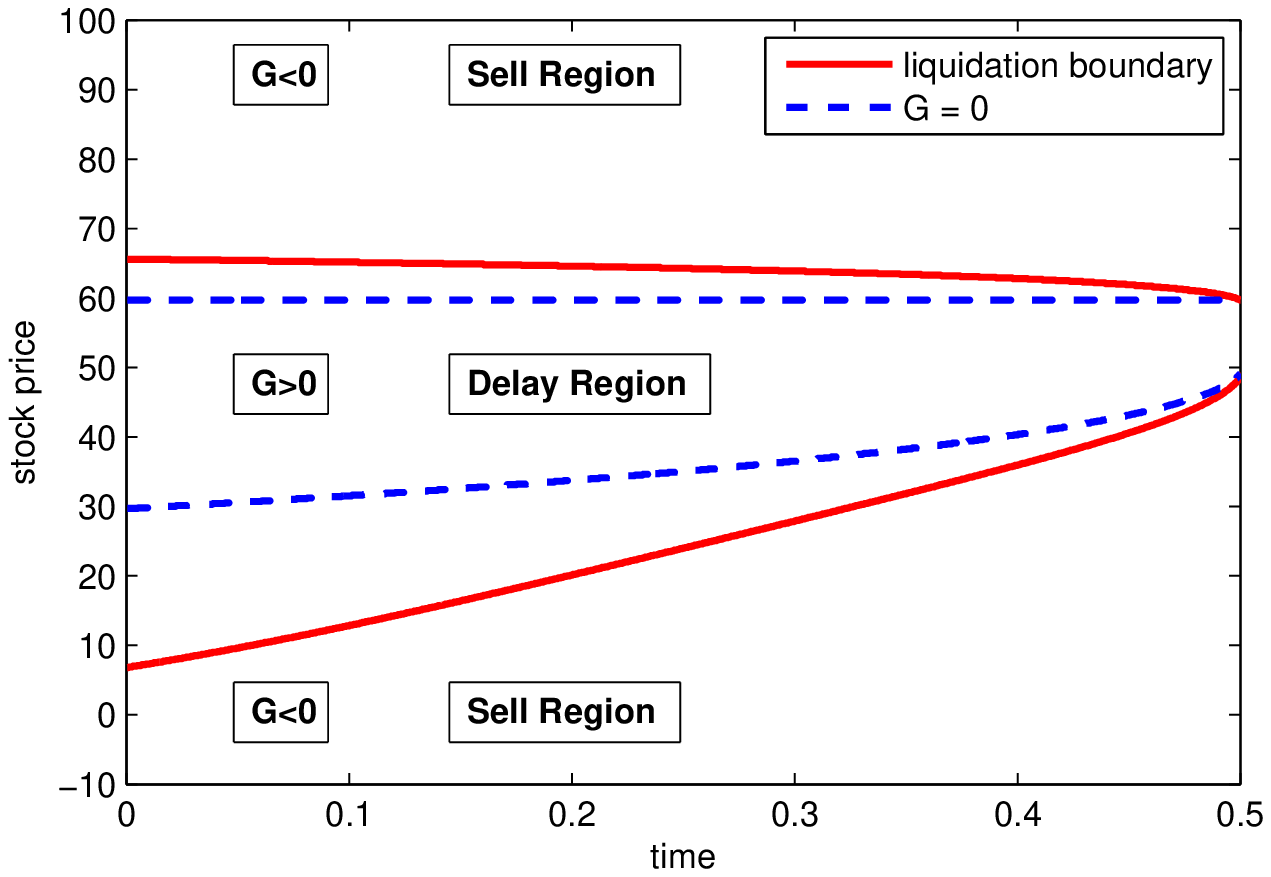}\qquad ~~
\includegraphics[scale=0.55]{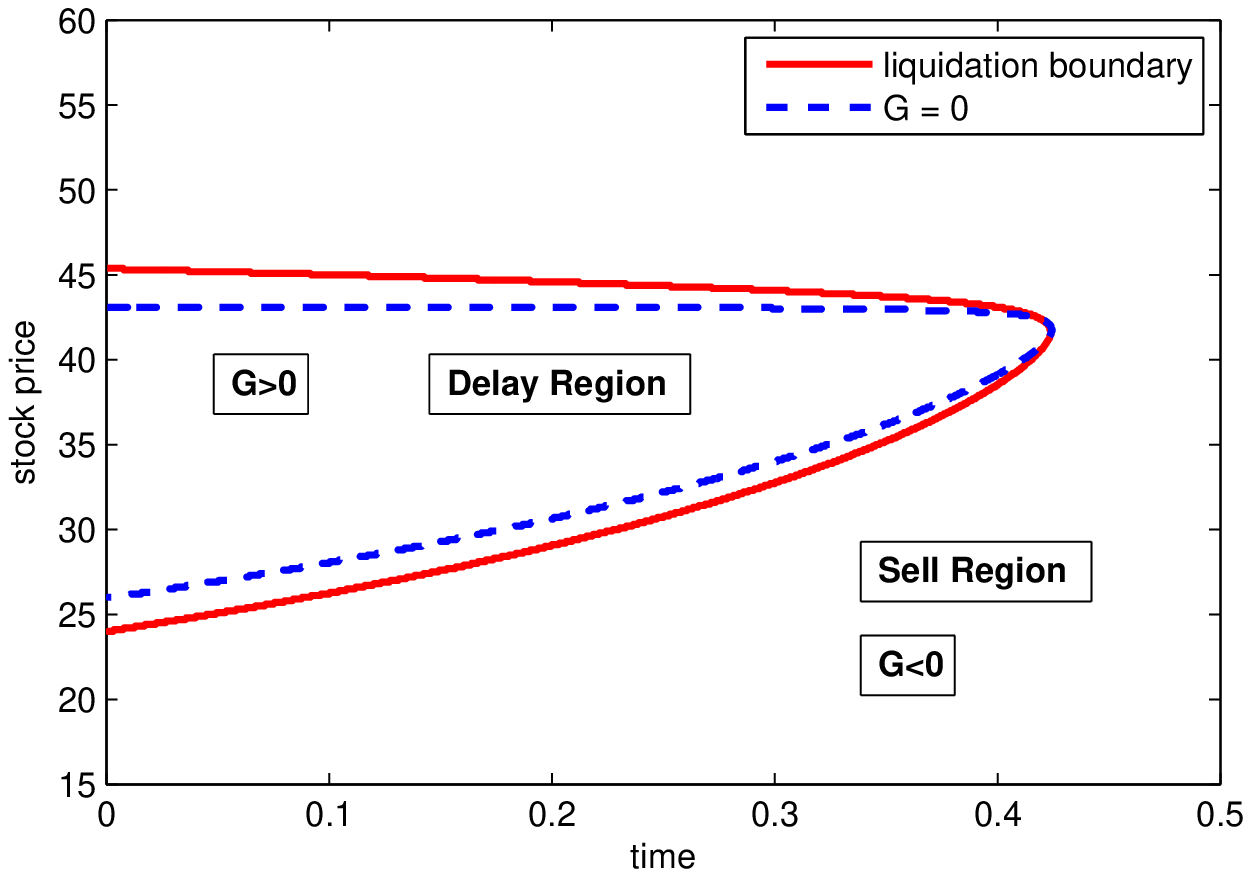}
\caption{\small{The liquidation boundary (red solid) and the zero contour of $G^{\alpha}$ (dashed) for a call under exponential OU dynamics. We take $\alpha=0.2$, $\theta=\log(60)$, $\beta=4$ in the left panel, and $\alpha=0.001$, $\theta=\log(50)$ and $\beta=0.2$ in the right panel, with common parameters $T=0.5$, $r=0.03$, $\sigma=0.3$, $\psi(\ell)=\ell$.}}
\label{Call_OU}
\end{figure}
In the right panel, we observe the interesting phenomena where the sell region is connected and contains the nonempty delay region. If the parameter $\beta$ (which measures the speed of mean reversion) is not sufficiently high, there may be no time for the price of the option to revert back to its long-term mean  before expiration, so that selling immediately becomes optimal close to maturity.

 \begin{proposition} \label{Put_OU_prop}
For the liquidation of a put option under the exponential OU model, the delay region is bounded if and only if $\alpha>0$.
\end{proposition}
\begin{proof}
The drive function is given by
\begin{equation}
G^{\alpha}_{Put}(t,s)=[r-\beta(\theta-\log(s))]s\Phi(-d_1)-\alpha\psi((m-P(t,s))^+),
\end{equation}
  If  $\alpha=0$, then we have $\{G^{\alpha}_{Put}>0\}=\{s>\exp(\frac{r}{\beta}-\theta)\}$. By \eqref{GsubL},  the delay region contains this set, so it is unbounded. 

Now let  $\alpha>0$, and we have the limit
\begin{equation}
\lim_{s\to \infty} G^\alpha_{Put}(t,s)=-\alpha\psi(m) < 0. \label{OUPut2}
\end{equation}
Next, we fix any $\hat{b}\in (0,\alpha \psi (m))$ and define $\overline{\psi}(\ell):=\min\{\psi(\ell),\hat{b}\}$. With this, we have 
\begin{equation*}
\overline{G}^\alpha(t,s):=[r-\beta(\theta-\log(s))]s\Phi(-d_1) - \alpha \overline{\psi}\left((m-P(t,s))^+\right) \geq G^\alpha_{Put}(t,s).
\end{equation*}
 We observe that $\overline{G}^\alpha$ is bounded above and  by \eqref{OUPut2} $\lim_{s\to\infty}\overline{G}^\alpha(t,s)\to -\alpha \hat{b} <0$ for every $t\in[0,T]$. Moreover, there exists $\hat{s}>0$ such that for every $s>\hat{s}$, $\overline{\psi}((m-P(t,s))^+)=\hat{b}$. As a result, we have
\begin{equation*}
\frac{\partial \overline{G}^\alpha}{\partial t}
= (\beta(\theta-\log(s))-r)s\phi(d_1)\frac{\log(\frac{s}{K}-(r+\frac{\sigma^2}{2})(T-t)}{2\sigma(T-t)^\frac{3}{2}}\leq 0,
\end{equation*} for $s>\max\{\hat{s}, \exp(\frac{r}{\beta}-\theta), K\exp((r+\sigma^2/2)T)\}$ and $t\in[0,T]$.  Also, we notice that $\overline{G}^\alpha (0,s)\to -\alpha \hat{b}$ as $s\to-\infty$. This allows us to  choose a  $b\in(0,\alpha \hat{b})$, then there exists $k>\max\{\hat{s}, K\exp((r+\sigma^2/2)T)\}$ such that $-b>\overline{G}(0,s)>\overline{G}(t,s)$ in for $(t,s) \in [0,T]\times [k,\infty)$. Therefore,  $\overline{G}$ satisfies the assumptions of Theorem \ref{compactsupport}. By Corollary \ref{Compare_G},  we conclude  the boundedness of the delay region. 
\end{proof}

Proposition \ref{Put_OU_prop} is illustrated in Figure \ref{Put_OU}.  When  $S_t$ is low, it is expected to revert back to the  (higher)  long-term mean, and the put price will decrease.  This generates an incentive to sell at a low stock price level. If $\alpha=0$,  when $S_t$ is high, there is no reason to sell since  the put price  is very low and expected to increase. Consequently, the delay region is on top of the sell region (Figure \ref{Put_OU}, left panel).
\begin{figure}[!ht]
\centering
\includegraphics[scale=0.55]{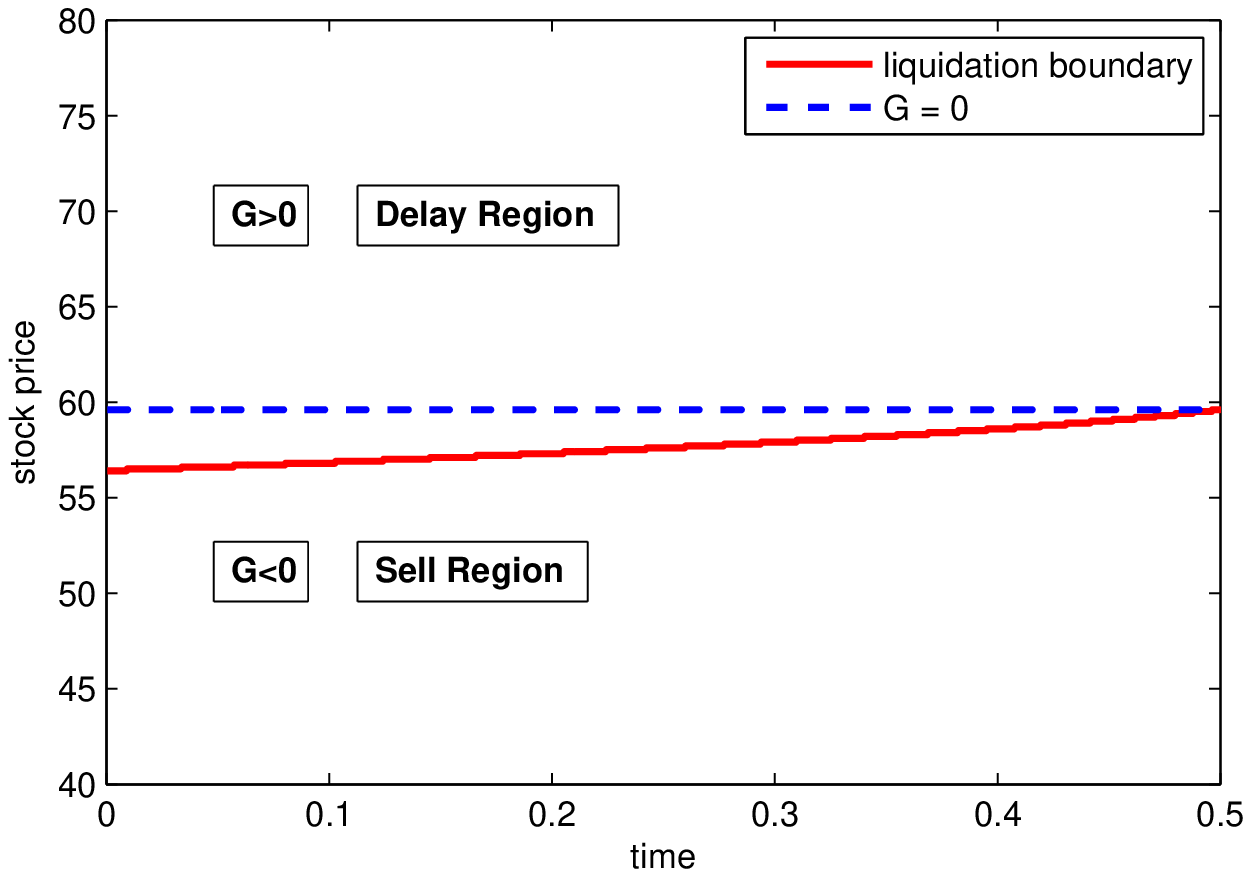}\qquad ~~
\includegraphics[scale=0.55]{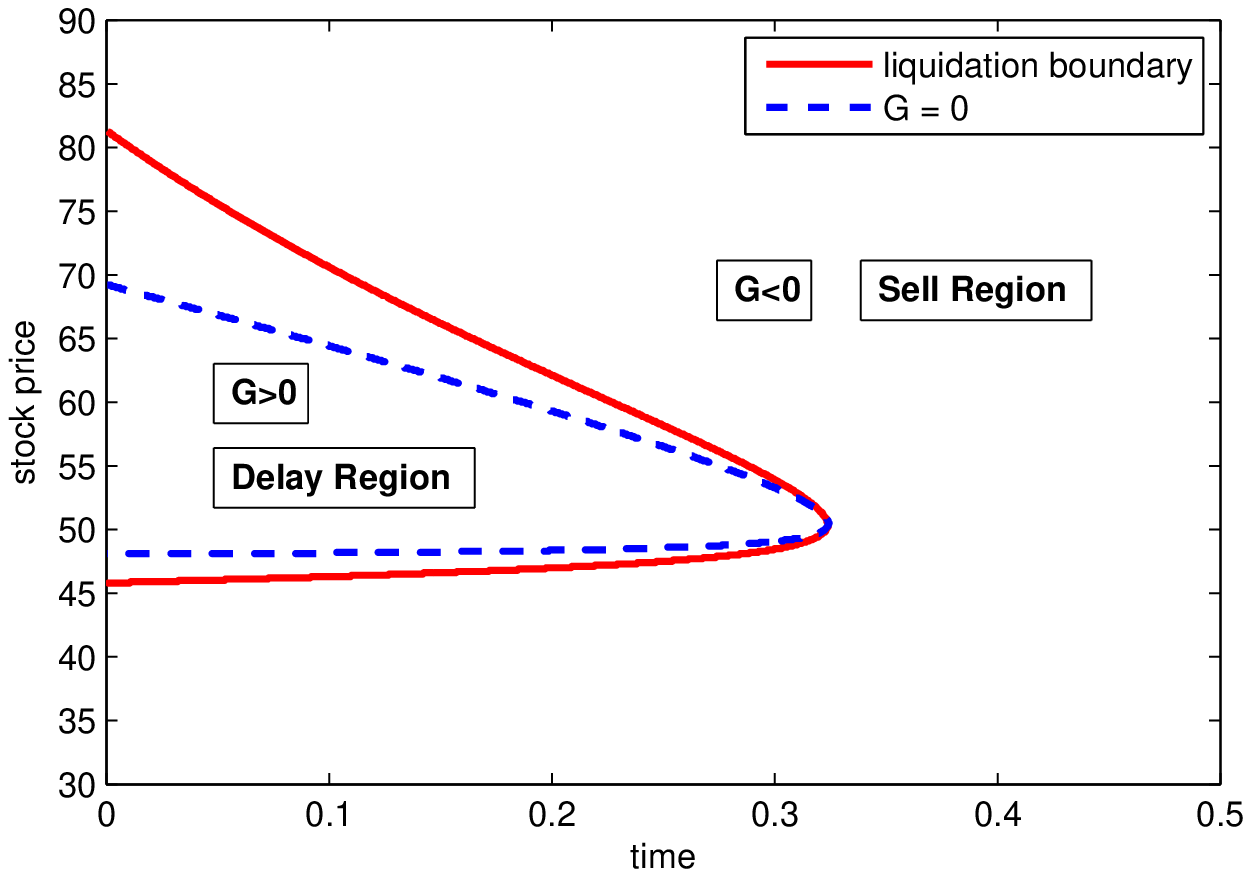}
\caption{\small{The liquidation boundary (solid) and the zero contour of $G^{\alpha}$ (dashed) for a put option under the exponential OU model. We take $\alpha=0$ and $K=50$ in the left panel, and $\alpha=0.01$ and $K=40$ in the right panel. Common parameters: $T=0.5$, $r=0.03$, $\sigma=0.3$, $\beta=4$ and $\theta=\log(60)$, $\psi(\ell)=\ell$. }}
\label{Put_OU}
\end{figure}
However, this is no longer true when we incorporate a non-zero risk penalty which reduces the value of waiting. As a result, the holder may sell the put at  high and low stock prices. In fact, if the penalization coefficient is large  and/or when the time-to-maturity is very short, the optimal liquidation premium may be zero at all stock price levels, resulting in an empty delay region (Figure \ref{Put_OU}, right panel).

\section{Quadratic Penalty} \label{sect:quadr}
As a variation to the shortfall-based penalty,  we consider a risk penalty based on the realized variance of the  option price process  from the starting time up to the liquidation time. Precisely, the investor now faces  the penalized optimal stopping problem
\begin{align*}
\tilde{J}^{\alpha}(t,s)&:= \sup_{\tau \in \mathcal{T}_{t,T}}\E_{t,s}\left\{e^{-r(\tau-t)}V(\tau,S_\tau)
							-\alpha \int_t^{\tau} e^{-r(u-t)}d[V,V]_u\right\} \\
						&= \sup_{\tau \in \mathcal{T}_{t,T}}\E_{t,s}\left\{e^{-r(\tau-t)}V(\tau,S_\tau)
							-\alpha \int_t^{\tau} e^{-r(u-t)} \sigma^2(u,S_u) S_u^2 V_s^2(u,S_u)du\right\},
\end{align*}
where $[V,V]$ denotes the quadratic variation of option price process $V$ defined in \eqref{mkt_pr}. Figure \ref{quadratic} illustrates the realized quadratic penalty associated with  a simulated call  option price path. Compared to the shortfall penalty in Figure \ref{shortfall}, the realized  quadratic penalty is increasing at all times, even when the option price is above its initial price.

\begin{figure}[ht]
\centering
\includegraphics[scale=0.65]{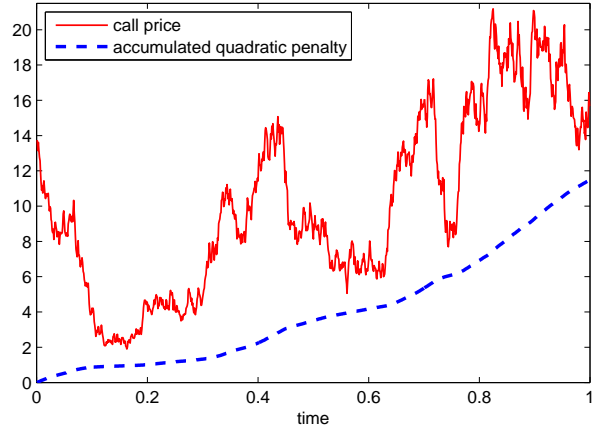}
\caption{\small{Realized quadratic penalty (dashed) based on a simulated price path (solid) of a call   under the GBM model with $\alpha =0.05$. The price path and other  parameters are the same as Figure \ref{realized_shortfall}. }}
\label{quadratic}
\end{figure}

Following  \eqref{delayed_liquidation}, we define the optimal liquidation premium by  $\tilde{L}^{\alpha}(t,s) := \tilde{J}^{\alpha}(t,s)-V(t,s)$. Again, we shall discuss  the stock or option liquidation problems  under the  GBM and exponential OU models. 

\subsection{Optimal Timing to Sell a Stock }\label{AnalyticalSolution}
We first  consider  the liquidation of a stock with the GBM  dynamics in terms of the  perpetual optimal stopping problem:
\begin{equation} \label{prob_repr_inf}
\tilde{L}^\alpha(s):=\sup_{\tau \in \mathcal{T}} \E_s\left\{ \int_0^\tau e^{-ru} \tilde{G}^{\alpha}(S_u) \,du \right\},
\end{equation}
with the  drive function  $\tilde{G}^\alpha(s):=(\mu-r)s-\alpha\sigma^2s^2$. If $\mu\leq r$, then selling immediately is always optimal since $\tilde{G}^\alpha$ is always negative. In contrast  if   $\mu>r$, then we obtain  a non-trivial closed-form solution. 

\begin{theorem}\label{perpL}
Let $\mu>r$. The value function $\tilde{L}^\alpha(s)$ in \eqref{prob_repr_inf} is  given by the formula
\begin{equation}\label{PerpGBM_solution}
\tilde{L}^\alpha(s)=\left\{\frac{\left(s^*\right)^{1-\lambda}}{2-\lambda}\,s^{\lambda}-s+B\,s^{2}\right\}\indic{s \,\leq \,s^*},
\end{equation}
where
\begin{align}
B & = \frac{\alpha\sigma^2}{2\mu+\sigma^2-r}, \quad \lambda = \frac{1}{\sigma ^{2}}\left[\, \frac{\sigma ^2}{2} - \mu +  \sqrt{\left(\frac{\sigma ^{2}}{2}-\mu \right)^{2}+2r\sigma^{2}}\,\right],\label{parameterB} \\
s^{\ast} & = \frac{1-\lambda}{( 2-\lambda)B}, \label{threshold}
\end{align}
and the stopping time $\tau^*=\inf\{t\geq0:S_t \ge s^*\}$ is optimal for \eqref{prob_repr_inf}.
\end{theorem}
\begin{proof}
We first show that \eqref{PerpGBM_solution} is the solution of
\begin{equation*}
\min \left\{r\Lambda(s)-\mu s\Lambda'(s)-\frac{\sigma ^2 s^2}{2}\Lambda''(s)-\tilde{G}^{\alpha}(s),\,\Lambda(s)\right\}=0, \quad s>0,
\end{equation*}
with $\Lambda(0)=0$. To do this, we split $\R^+$ into two regions: $\mathcal{D}_1=(0,s^*)$ and $\mathcal{D}_2=[s^*,\infty)$ with $s^*>0$ to be determined. We conjecture that $\Lambda(s)=0$ in $\mathcal{D}_2$, and  for $s\in \mathcal{D}_1$ $\Lambda(s)$ solves
\begin{equation}\label{PerpGBM_PDE}
r\Lambda(s)-\mu s\Lambda'(s)-\frac{\sigma^2s^2}{2}\Lambda''(s)-\tilde{G}^\alpha(s)=0.
\end{equation}
By direct substitution, the general solution to equation \eqref{PerpGBM_PDE} is of the form
\begin{equation*}
\Lambda(s)=C_{1}s^{\lambda_{1}}+C_{2}s^{\lambda _{2}}-s+Bs^2,
\end{equation*}
where $C_{1}$ and $C_{2}$ are constants to be determined, $B$ is specified in \eqref{parameterB} and 
\begin{equation*}
\lambda_{k} =  \frac{1}{\sigma ^{2}}\left[\, \frac{\sigma ^2}{2} - \mu 				+(-1)^k\sqrt{\left(\mu -\frac{\sigma ^{2}}{2}\right)^{2}+2r\sigma^{2}}\,\right], \quad k \in \{1,2\}.
\end{equation*}
We apply the continuity and smooth pasting conditions at $s=0$ and $s=s^*$ to get
\begin{align}
&\lim_{s\downarrow 0}\Lambda(s)=0 ~\Rightarrow ~C_1=0,\notag\\
&\lim_{s\uparrow s^*}\Lambda(s)=0 ~\Rightarrow~ C_2(s^*)^{\lambda_2}-s^*+B(s^*)^2=0,\label{Lcont}\\
&\lim_{s\uparrow s^*}\Lambda'(s)=0~ \Rightarrow~ \lambda_2C_2(s^*)^{\lambda_2-1}-1+2Bs^*=0. \label{smoothL}
\end{align}
Solving the system of equations   \eqref{Lcont}--\eqref{smoothL} gives $C_2$ and $s^*$ as in \eqref{parameterB}-\eqref{threshold}. One can verify by substitution that $\Lambda(s)$ is indeed a  classical solution of \eqref{PerpGBM_PDE}.

By Ito's formula and \eqref{PerpGBM_PDE}, $(\Lambda(S_t))_{t\ge 0}$ is a $(\P, \Fil)$-supermartingale,  so  for every $\Fil$-stopping time $\tau$ and $n\in\N$, we have 
\begin{equation} \label{F-K}
\Lambda(s)\geq \E_{0,s} \left\{\int_0^{\tau\wedge n} e^{-ru}\tilde{G}^\alpha(S_u)du\right\}.
\end{equation}
Maximizing \eqref{F-K} over $\tau$ and $n$ yields that  $\Lambda(s)\geq \tilde{L}^\alpha(s)$ for $s \ge 0$. The reverse inequality is deduced  from  the probabilistic representation $\Lambda(s) =  \E_{0,s} \left\{\int_0^{\tau^*} e^{-ru}\tilde{G}^\alpha(S_u)du\right\}$, with the candidate stopping time   $\tau^*:=\inf\{t\geq0\,:\,S_t \ge s^*\}$. Hence, we conclude that $\Lambda(s)=\tilde{L}^\alpha(s)$ and $\tau^*$ is optimal. \end{proof}

The optimal liquidation threshold $s^{\ast}~$ in \eqref{threshold} is non-negative if and only if  $\lambda_{2}<1$, which is equivalent to the  condition $\mu>r$ in Theorem \ref{perpL}. Otherwise,  $\tilde{L}^\alpha(s)=0$ and  the optimal strategy is  to sell immediately.

In Figure \ref{PerpGBM_Fig1}, we illustrate the optimal liquidation premium   $\tilde{L}^\alpha(s)$ for various values of $\mu$ and $\sigma$.  As $\mu$ increases,  the optimal threshold as well as the optimal  liquidation premium (at all stock price levels) increase (left panel). On the other hand, a higher volatility reduces the optimal  liquidation premium at every initial stock price. We also observe that $\tilde{L}^\alpha(s)$  smooth-pastes the level 0 at the optimal threshold $s^*$, as is expected from \eqref{Lcont} and \eqref{smoothL}.

\begin{figure}[ht]
\begin{center}
\includegraphics[scale=0.55]{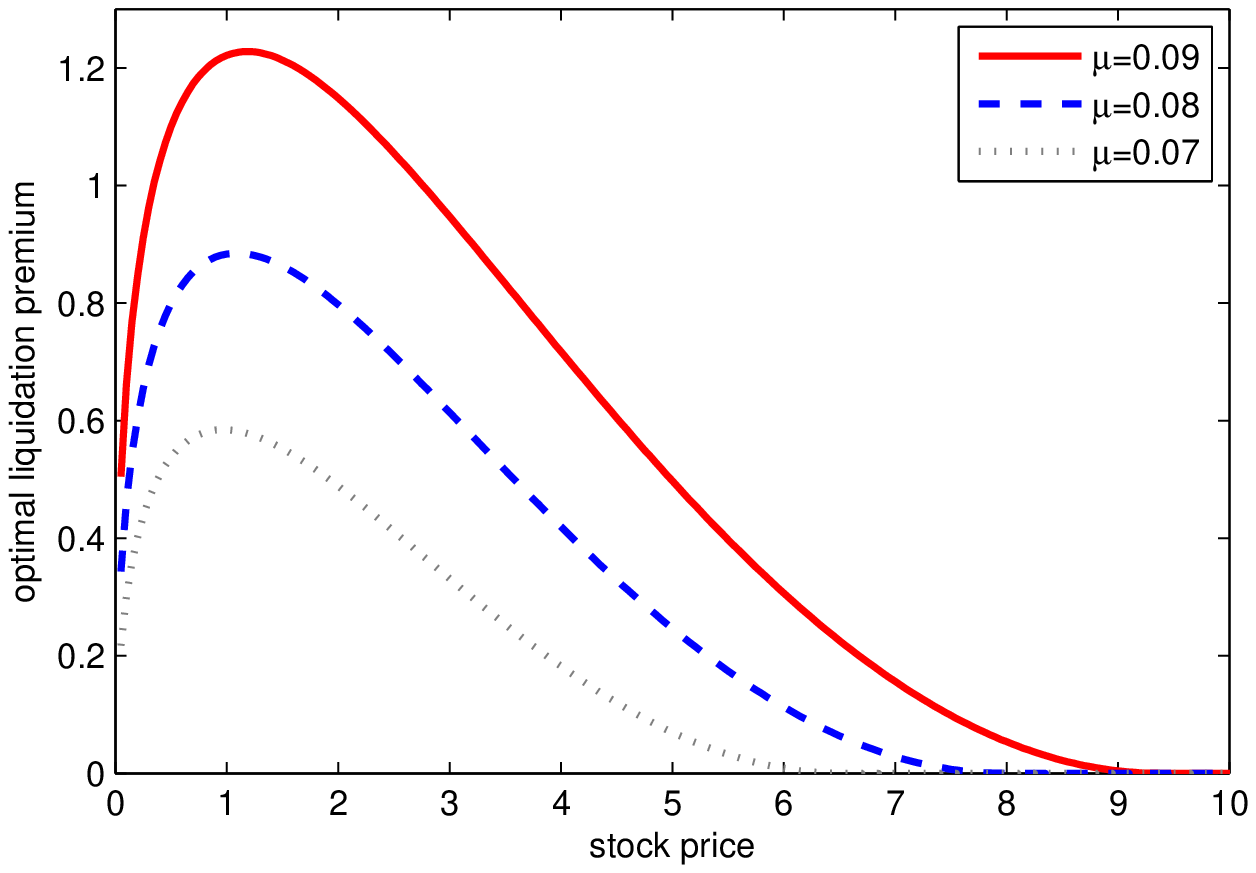}\qquad ~~
\includegraphics[scale=0.55]{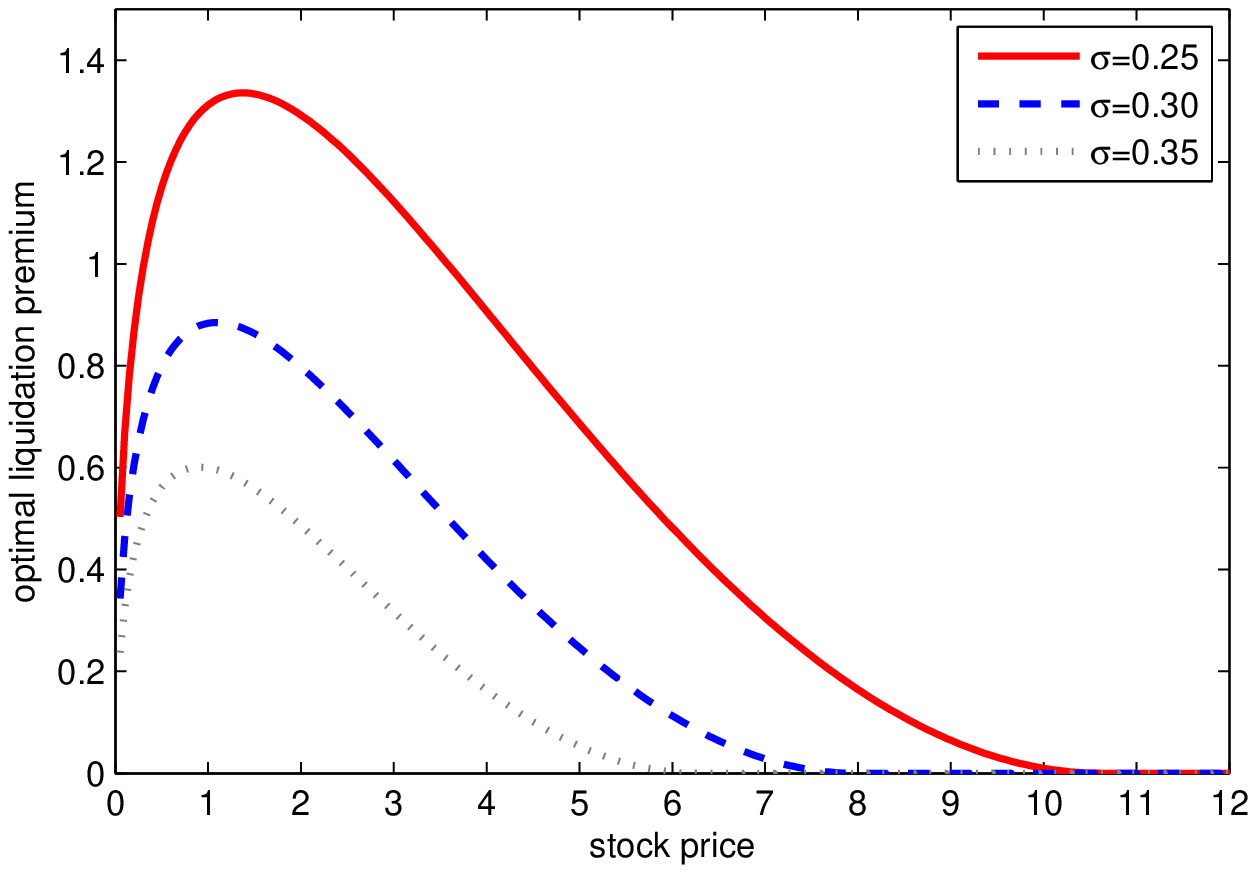}
\end{center}
\caption{\small{The optimal liquidation premium for a stock under the GBM model  for different values of $\mu$ and $\sigma$. In the left panel, we take $r=0.03$, $\sigma=0.3$ and $\alpha=0.2$, and  the liquidation threshold $s^* = 9.37, 7.97,  6.52$ for $\mu = 0.09, 0.08, 0.07$ respectively. In the right panel, we take $r=0.03$,  $\mu=0.08$,  and $\alpha=0.1$, and the liquidation threshold  $s^* = 10.63, 7.97,  6.26$ for $\sigma=0.25, 0.30, 0.35$. }}
\label{PerpGBM_Fig1}
\end{figure}

If $S$ follows the exponential OU dynamics, the drive function for liquidating a stock is 
\begin{equation}\label{G_PerpStock_OU}
\tilde{G}^\alpha(s)=[\beta(\theta-\log(s))-r-\alpha s]s.
\end{equation}
In this case, we do not have a closed-form solution. Nevertheless we observe from \eqref{G_PerpStock_OU} that the delay region is non-empty, namely, $\{\tilde{L}^\alpha>0\}\supseteq\{s<\tilde{s}\}$, where $\tilde{s}$ is determined uniquely from the equation
\begin{equation}\notag
\beta(\theta-\log(\tilde{s}))-r-\alpha \tilde{s} = 0.
\end{equation}
On the other hand, since $\tilde{G}^\alpha\to-\infty$ as $s\to\infty$, we expect intuitively that the investor will sell when the stock price is high.

\subsection{Liquidation of Options}
We now discuss some numerical  examples to demonstrate the  liquidation strategies for  European call and put options. With  strike  $K$ and maturity $T$, the drive functions are respectively given by
\begin{align}
\tilde{G}^{\alpha}_{Call}(t,s) & = s\Phi(d_1)\big(\mu-r-\alpha \sigma^2 s \Phi(d_1)\big), \label{call_quadr_gbm} \\
\tilde{G}^{\alpha}_{Put}(t,s) & = s\Phi(-d_1)\big(r-\mu-\alpha \sigma^2 s\Phi(-d_1)\big). \label{put_quadr_gbm}
\end{align}

 When $\mu\leq r$ and $\alpha>0$, the drive function $\tilde{G}^{\alpha}_{Call}(t,s)$ is negative for all $(t,s)$, so it is optimal to sell the call immediately. However, when  $\mu>r$ and $\alpha>0$, we notice from \eqref{call_quadr_gbm} that, when the stock price is sufficiently large (resp. small), the drive function of a call is negative (resp. positive). Hence, as we see in Figure \ref{fig_Quadr_GBM}, it is optimal to sell the call when the stock price is high, and the optimal liquidation boundary is lower as the penalization coefficient increases.  In contrast to the shortfall penalty,  the investor now  is  subject to a higher penalty when the stock price is high under the quadratic penalty. Consequently, the sell region is now above  the delay region, as opposed to being at the bottom in the shortfall case in Figure  \ref{Stock_Call_SF_1} (right panel). 

In the put option case,  we observe from \eqref{put_quadr_gbm} that
\begin{equation*}
\lim_{s\rightarrow 0}r-\mu-\alpha \sigma^2 s\Phi(-d_1)=\lim_{s\rightarrow \infty}r-\mu-\alpha \sigma^2 s\Phi(-d_1)=r-\mu.
\end{equation*}
Consequently, when $\mu<r$ and the stock price is sufficiently large or small, the drive function is strictly positive and it is optimal to hold the position. In contrast,    the shortfall converges to $\psi(m)>0$ as $s$ increases (see \eqref{Put2}), which means that  it is optimal to sell when the stock price is high  (see Figure \ref{Put_SF}).  We illustrate the timing strategies under quadratic penalty  in Figure \ref{fig_Quadr_GBM} (right). As expected there is a low and a high delay regions which are separated by a sell region in the middle. Also we notice that as the penalization coefficient $\alpha$ increases, the sell region expands.   

\begin{figure}[ht]
\begin{center}
\includegraphics[scale=0.55]{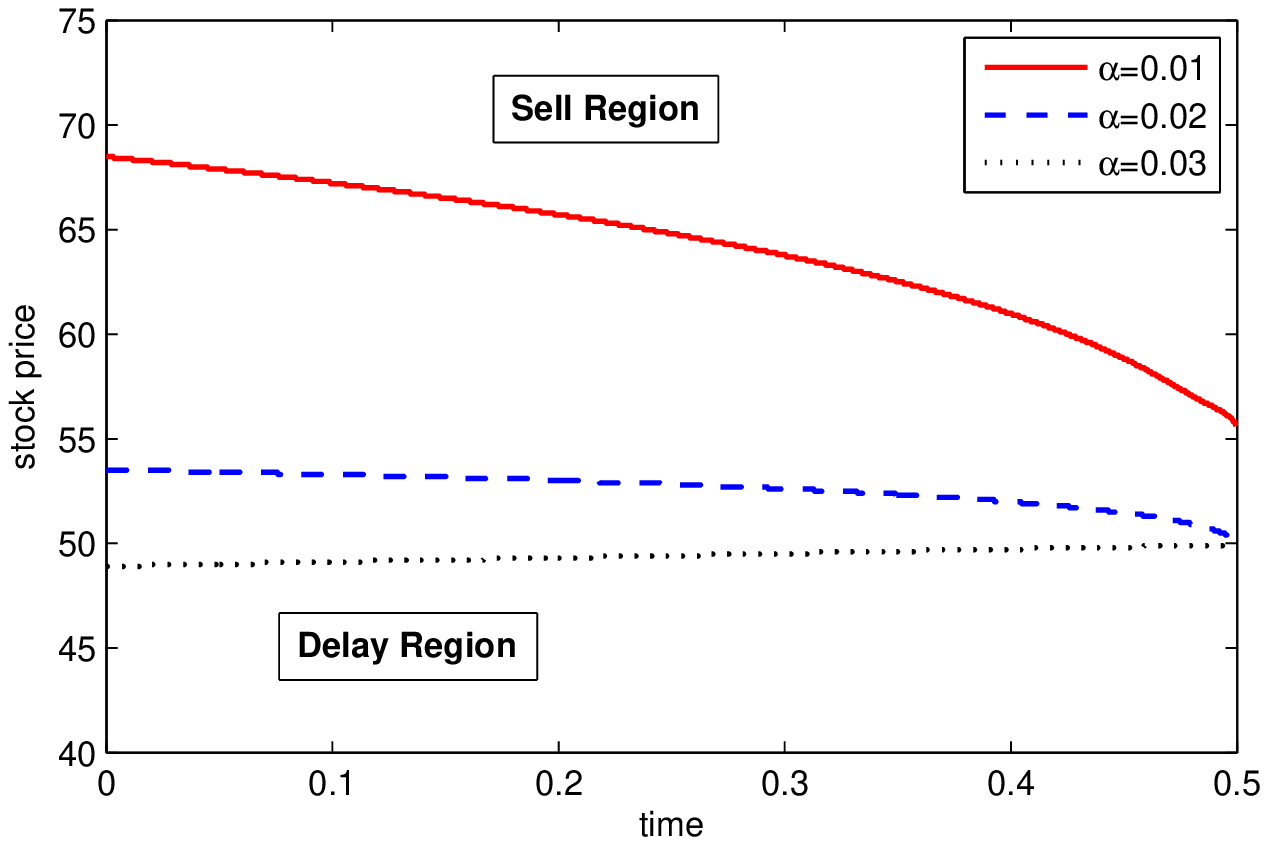}
\includegraphics[scale=0.55]{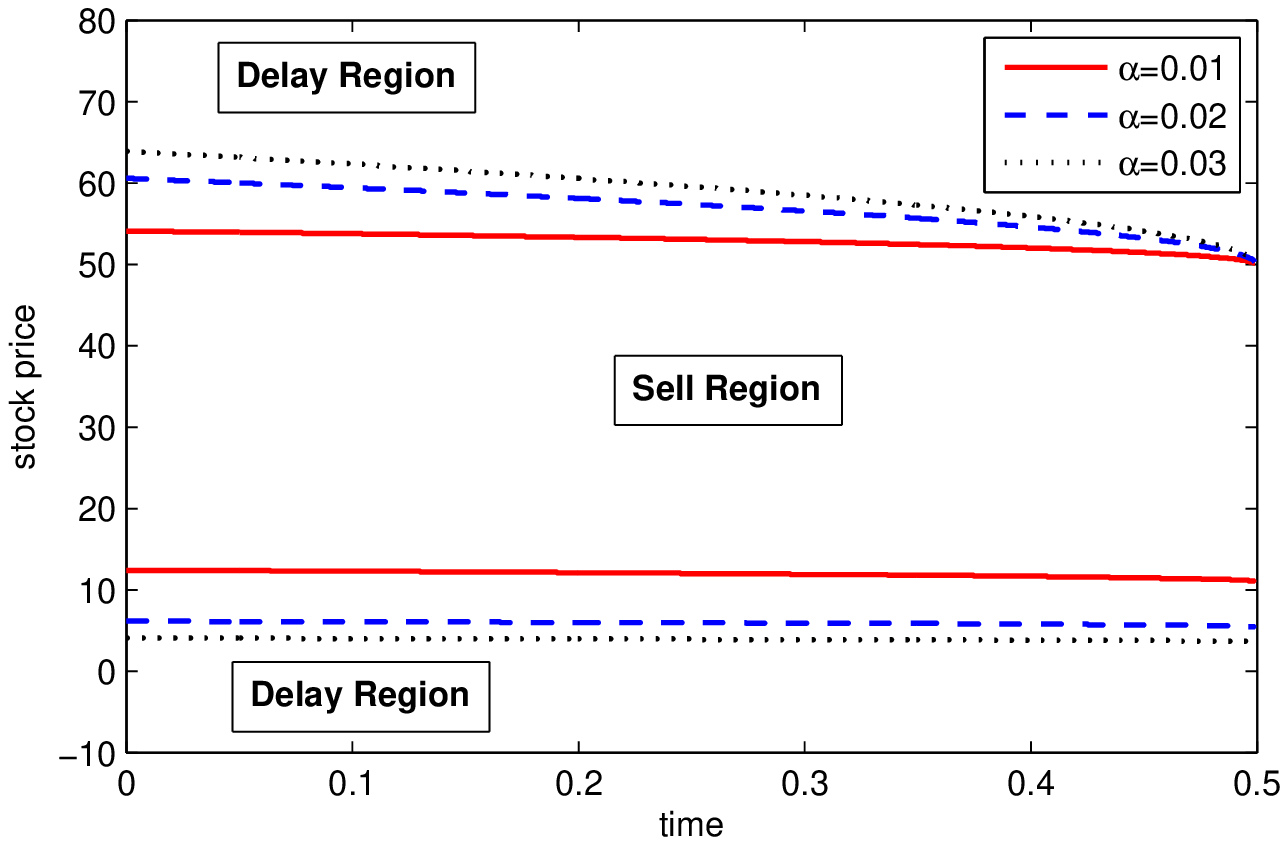}
\end{center}
\caption{\small{The liquidation boundaries for a call option (left panel) and a put option (right panel) under the GBM mode with  different values of  $\alpha$. Parameters: $T=0.5$, $r=0.03$, $\sigma=0.3$, $K=50$, $\mu=0.08$ (call) and $\mu=0.02$ (put). }}
\label{fig_Quadr_GBM}
\end{figure}

Under the exponential OU model, the drive functions for selling a call   and a put  are, respectively,
\begin{align*}
\tilde{G}_{Call}^\alpha(t,s) 
		&= s\Phi(d_1) \big( \theta - r - \beta \log s - \alpha \sigma^2 s \Phi(d_1) \big) ,\\
\tilde{G}^{\alpha}_{Put}(t,s)
		&= s\Phi(-d_1)\bigg(r-\theta + \beta \log s-\alpha \sigma^2 s\Phi(-d_1)\bigg).
\end{align*}
In Figure \ref{fig_Quadr_OU},  we can visualize the optimal liquidation premium $\tilde{L}^\alpha(t,s)$ for a call (right panel) and a put (left panel). In the call case, the delay region,  which corresponds to the area where $\tilde{L}^\alpha>0$,  is bounded. When $s$ is sufficiently high, $\tilde{L}^\alpha$ vanishes and it is optimal to sell. This is intuitive since $\lim_{s\to\infty}\tilde{G}_{Call}^\alpha(t,s)=-\infty$ and $\tilde{G}_{Call}^\alpha$ is positive for sufficiently small $s$.

In contrast, the drive function for the put $\tilde{G}^{\alpha}_{Put}(t,s)$  is negative when $s < \exp\left(\frac{\theta - r}{\beta}\right)$, for every $\alpha \ge 0$. Therefore, as Figure \ref{fig_Quadr_OU} indicates, one expects the optimal liquidation premium to vanish for small $s$, so the investor will sell when the   put price is high.  
 Compared to Figure  \ref{Put_OU} with a shortfall penalty, the investor does not sell when the underlying stock price is very high.  This is because the drive function  $\tilde{G}^{\alpha}_{Put}(t,s)$  stays positive for large $s$ (recall  \eqref{GsubL}).  As time approaches maturity, the delay liquidation premium   decreases  to  its terminal condition of value zero. 

\begin{figure}[ht]
\begin{center}
\includegraphics[scale=0.55]{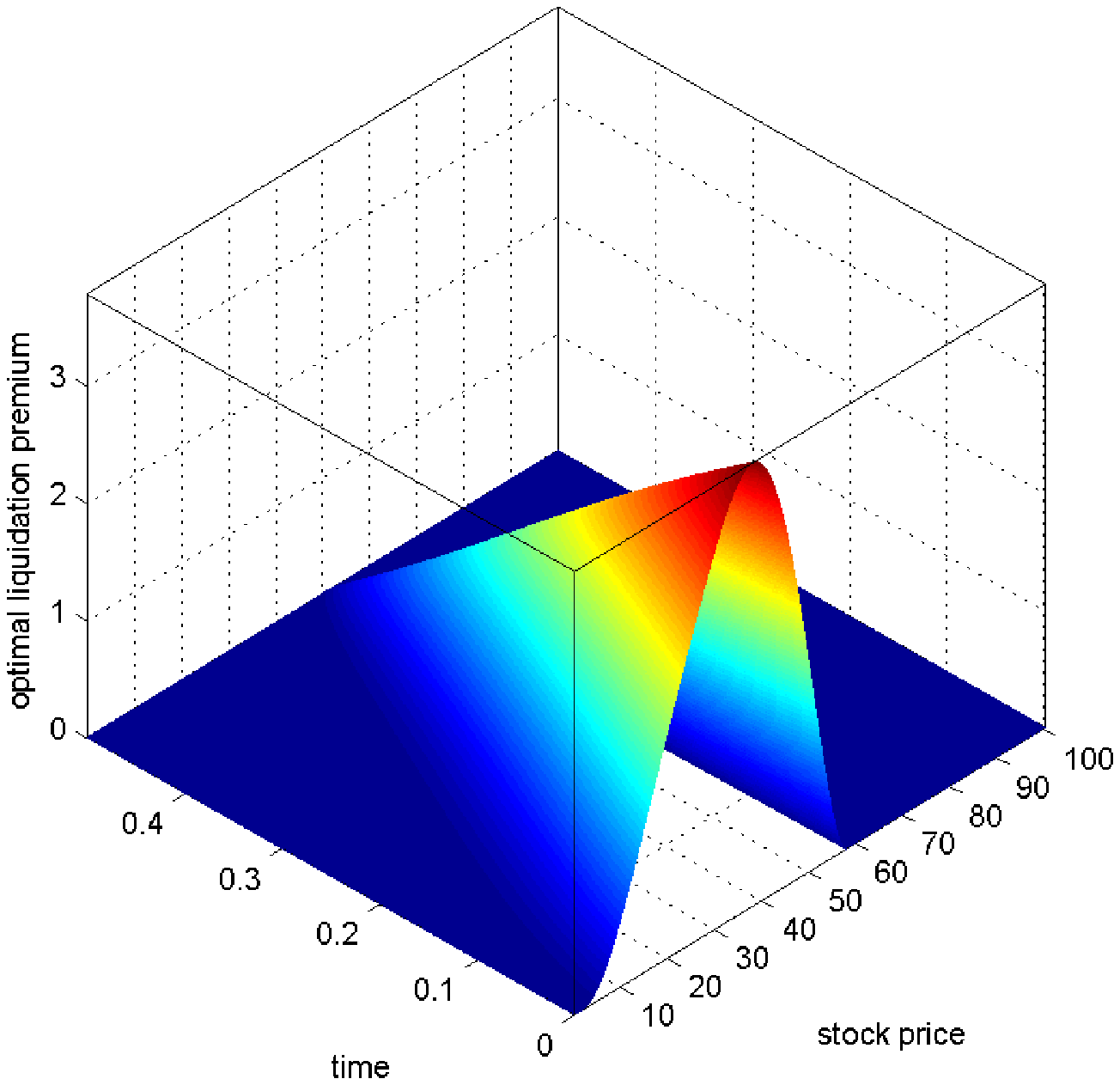}\qquad ~~
\includegraphics[scale=0.55]{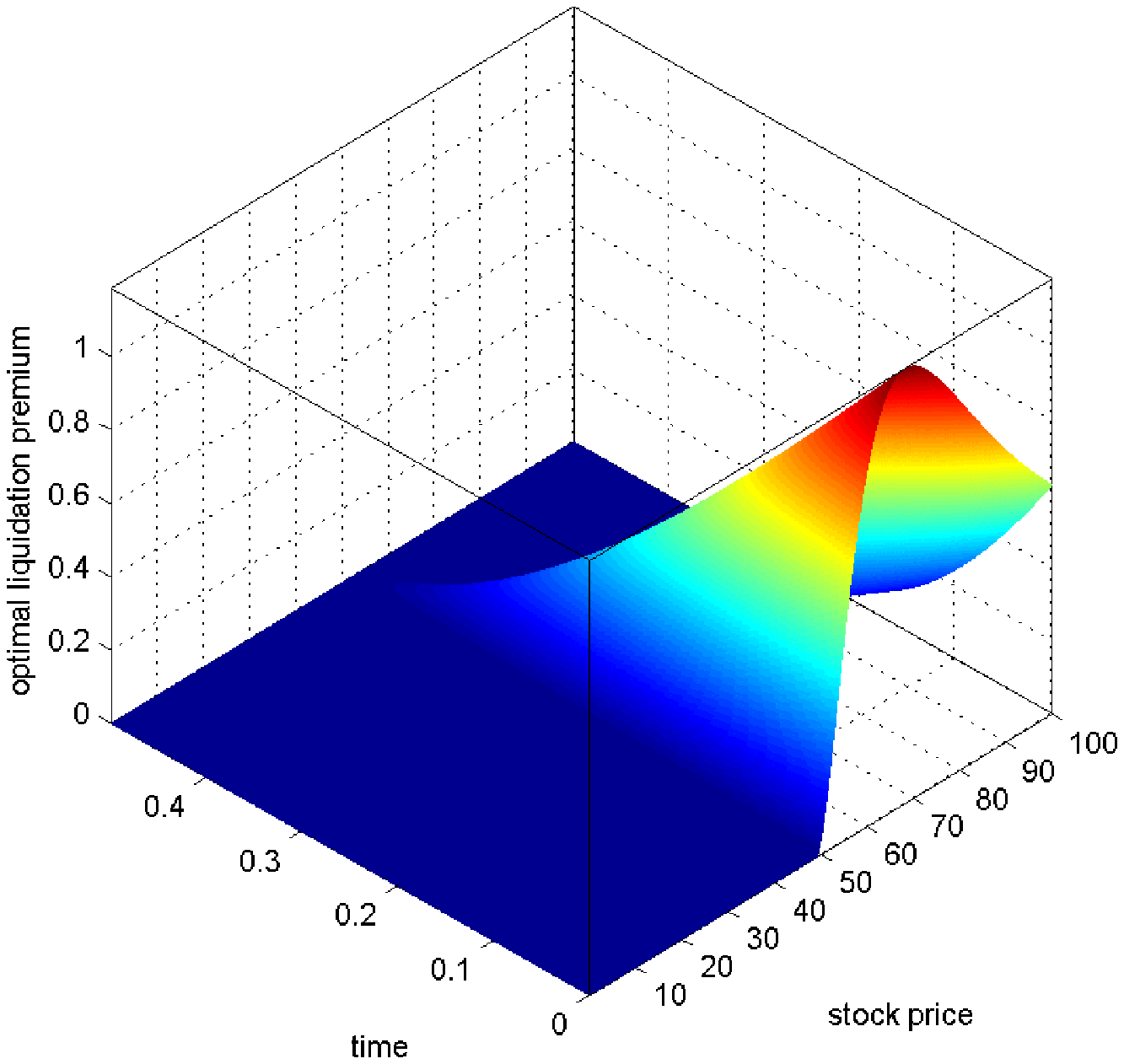}
\end{center}
\caption{\small{The optimal liquidation premium for a call option (left) and a put option (right) with exponential OU dynamics. We take $T=0.5$, $r=0.03$, $\sigma=0.3$, $K=50$, $\alpha=0.1$, $\beta=4$ and $\theta=\log(60)$. }}
\label{fig_Quadr_OU}
\end{figure}

\section{Conluding Remarks} \label{sect:concl}
In summary, we have provided a flexible mathematical model for the optimal liquidation of option positions under a path-dependent penalty. We have identified the situations where the optimal timing is trivial, and solved for non-trivial liquidation strategy via variational inequality. The  penalty type as well as the  penalization  coefficient can give rise to very different   liquidation timing. Our findings are useful for both individual and institutional investors who use options for speculative  investments or risk management purposes.

For future research, a natural direction is to  adapt our model to the  problem of  sequentially buying and selling  an option.  Moreover, one can consider applying the methodology to  derivatives other than equity options. For example, we refer to  \cite{LeungLiu2012} for a recent study on   the liquidation  of credit derivatives with  pricing measure discrepancy but without risk penalty. It would be both mathematically interesting  and challenging to study  option  liquidation under incomplete markets. On the other hand, our model can be extended to markets with liquidity cost and price impact (see e.g. \cite{Almgren2003,Lorenz11,Schied2009}).  Finally, the  path-dependent risk penalization can also be incorporated   to  dynamic  portfolio optimization problems to account for adverse performance    during  the  investment horizon.

\section{Strong Solution to the Inhomogeneous Variational Inequality}\label{Sect-ExUn}
In this section, we   follow the terminology and procedures in   \cite{Bensoussan78}, and establish   the existence and uniqueness of a strong solution  to the variational inequality \eqref{L_general_VI} under conditions that are applicable to the GBM and exponential OU models.

\noindent \textbf{Preliminaries.}  We express prices in logarithmic scale by setting $X_t=\log(S_t)$. Equation \eqref{SDE_S} then becomes
\begin{equation} \label{SDE_X}
dX_t = \eta(t,X_t)dt + \kappa(t,X_t)dW_t,
\end{equation}
for some functions $\kappa(t,x)$ and $\eta(t,x)$. Next, we define the operator $\mathcal{A}$ by
\begin{align}
\mathcal{A}[\cdot] 
&= -\frac{\kappa^2(t,x)}{2} \frac{\partial^2\, \cdot}{\partial x^2}-\eta(t,x) \frac{\partial\, \cdot}{\partial x} + r \,\cdot \notag \\
&= -\frac{\partial\,\cdot}{\partial x}\left(a_2(t,x)\frac{\partial \,\cdot}{\partial x}\right) + a_1(t,x)\frac{\partial \,\cdot}{\partial x} + r\,\cdot  \,, \label{divergenge_form}
\end{align}
where
\begin{equation*}
a_1(t,x) = \frac{1}{2}\frac{\partial }{\partial x}\kappa^2(t,x)-\eta(t,x), \qquad
a_2(t,x) = \frac{\kappa^2(t,x)}{2}.
\end{equation*} In term of log-prices, we express  the drive function as  $g(t,x)=G^\alpha(t,e^x)$ and the optimal liquidation premium as $u(t,x) = L(t,e^x)$.  Throughout, we denote the domain  $\mathcal{D}=[0,T]\times \R$. In order to solve the VI \eqref{L_general_VI}, it is equivalent to solve the VI problem:
\begin{equation} \label{u_general_VI}
\begin{cases}
-\frac{\partial u}{\partial t} + \mathcal{A}[u] - g(t,x) \geq 0, \,u(t,x)\geq0, \quad (t,x) \in \mathcal{D}, \\
\left(-\frac{\partial u}{\partial t} + \mathcal{A}[u] - g(t,x)\right) u = 0, \quad (t,x) \in \mathcal{D}, \\
u(T,x)=0, \quad x \in \R.
\end{cases}
\end{equation}

We  describe an appropriate class of solutions for \eqref{u_general_VI} in a suitable Sobolev space and prove that such a solution exists and is unique.   First, let us  define, for   $\lambda(x)=\exp(-n|x|)$, $n\in \mathbb{N}$,
\begin{align*}
\mathcal{L}^2_{\lambda}(\R)   & = \{v\,|\,\sqrt{\lambda} v \in L^2(\R)\}, \\
\mathcal{H}^1_{\lambda}(\R)   & = \{v\in L^2_{\lambda}(\R)\,|\,\frac{\partial v}{\partial x} \in L^2_{\lambda}(\R)\},\\
\mathcal{H}^1_{0,\lambda}(\R) & = \{v\in H^1_{\lambda}(\R)\,|\,\lim_{|x|\to\infty}v(x)=0\}.
\end{align*} These are Hilbert spaces when endowed with the following inner products
\begin{align*}
(f,g)_{L^2} & = \int_{\R}\lambda fgdx, \,\quad  f,g \in L^2_{\lambda}(\R),\\
(f,g)_{H^1} & = \int_{\R}\lambda fgdx + \int_{\R}\lambda \frac{\partial f}{\partial x}\frac{\partial g}{\partial x}dx,\quad  f,g \in H^1_{\lambda}(\R).
\end{align*}

We denote by $\mathcal{H}^1_{c,\lambda}(\R)$ the set of functions $w \in \mathcal{H}^1_{\lambda}(\R)$ with compact support. For $u \in \mathcal{H}^1_{0,\lambda}(\R)$, $w \in \mathcal{H}^1_{c,\lambda}(\R)$, we define the operator
\begin{equation*}
\mathcal{I}_{\lambda}(t,u,w) = \int_{\R} a_2(t,x)\left(\lambda\frac{\partial u}{\partial x}\frac{\partial w}{\partial x}
  + w\frac{\partial u}{\partial x}\frac{\partial \lambda}{\partial x}\right)dx
  + \int_{\R} a_1(t,x)\lambda\frac{\partial u}{\partial x}w dx +r\int_{\R}\lambda uw dx.
\end{equation*}
We can assume without loss of generality  \citep[Sect. 3.2.17]{Bensoussan78} that $\mathcal{I}_{\lambda}$ is coercive on $\mathcal{H}^1_{c,\lambda}(\R)$, i.e.
\begin{equation*}
\mathcal{I}_{\lambda}(t,w,w)\geq\alpha |\!|w|\!|_{H^1} \quad \forall \, w \in \mathcal{H}^1_{c,\lambda}(\R), \, \alpha>0.
\end{equation*}
Integrating by parts allows us to extend $\mathcal{I}_{\lambda}$ to a bilinear form on the whole space $\mathcal{H}^1_{0,\lambda}(\R)$. In particular, we set
\begin{equation*}
\mathcal{I}_{\lambda}(t,u,v) = \int_{\R}\left[a_2(t,x)\lambda\frac{\partial u}{\partial x}\frac{\partial v}{\partial x}
  + a_2(t,x)\frac{\partial \lambda}{\partial x}\frac{\partial u}{\partial x}v\right]dx
  + \int_{\R} \left(r - \frac{1}{2}\frac{\partial a_1}{\partial x}
  - \frac{1}{2\lambda}a_1\frac{\partial \lambda}{\partial x} \right) \lambda uvdx.
\end{equation*}
with $u,v \in \mathcal{H}^1_{0,\lambda}(\R)$.

Following Chapter 5.9.2 of \cite{Evans1998} and Chapter 2.6 of \cite{Bensoussan78}, we define the  space $\mathcal{L}^p(0,T;X)$ consisting  of all strongly measurable functions $\chi:[0,T]\to X$ with
\begin{equation*}
|\!|\chi|\!|_{\mathcal{L}^p(0,T;X)}=\left(\int_0^T|\!|\chi(t)|\!|^p_X\,dt\right)^{1/p}, \quad 1\leq p<\infty,
\end{equation*}
and for $p=\infty$,
\begin{equation*}
|\!|\chi|\!|_{\mathcal{L}^{\infty}(0,T;X)}=\esssup_{0\leq t \leq T}|\!|\chi(t)|\!|_X,
\end{equation*}
For $\chi\in \mathcal{L}^1(0,T;X)$, we say $\nu\in \mathcal{L}^1(0,T;X)$ is the weak derivative of $\chi$, denoted by $\nu= \frac{\partial \chi}{\partial t}$, if
\begin{equation*}
\int_0^T\frac{\partial w}{\partial t}\chi(t) dt = -\int_0^T w(t)\nu(t)dt, \quad \forall \,w \in C_c^\infty([0,T]).
\end{equation*}
The Sobolev space $\mathcal{H}^1(0,T;X)$ consists of all functions $\chi \in \mathcal{L}^2(0,T;X)$ such that the weak derivative exists and belongs to $\mathcal{L}^2(0,T;X)$. Furthermore, we set
\begin{equation} \label{H1norm}
|\!|\chi|\!|_{\mathcal{H}^1(0,T;X)}=\left(\int_0^T|\!|\chi(t)|\!|^2_X+|\!|\frac{\partial }{\partial t}\chi(t)|\!|^2_X\,dt\right)^{1/2},
\end{equation}
which makes $\mathcal{H}^1(0,T;X)$ an Hilbert space (see Chapter 5.9.2 in \cite{Evans1998}).

\vspace{5 mm}

\noindent  \textbf{Main Results.}
\begin{definition}
A function $u:\mathcal{D}\to\R$ is a \emph{strong solution} of problem \eqref{u_general_VI} if, $\forall$ $v \in \mathcal{H}^1_{\lambda}(\R)$, $v\geq0$ a.e., the following conditions are satisfied:
\begin{equation}\label{L_general_VI_strong}
\begin{cases}
& u\in \mathcal{L}^2(0,T;\mathcal{H}^1_{0,\lambda}(\R)), \, \frac{\partial u}{\partial t} \in \mathcal{L}^2(0,T;\mathcal{L}^2_{\lambda}(\R)), \\
& -\left(\frac{\partial u}{\partial t},v-u\right)-\mathcal{I}_{\lambda}(t;u,v-u)\leq(g,v-u), \\
& u\geq0 \text{ a.e. in } \mathcal{D}, \\
& u(T,x)=0, \, x\in\R.
\end{cases}
\end{equation}
\end{definition}
We shall impose the following conditions on $a_2$, $a_1$, $g$.\\
\textbf{Assumption A}. $a_2$, $\frac{\partial a_2}{\partial t}$ and $\frac{\partial a_1}{\partial x}\in\mathcal{L}^{\infty}(\mathcal{D})$; $a_1$ and $\frac{\partial a_1}{\partial t} \in C^0(\overline{\mathcal{D}})$; $g \in \mathcal{H}^1(0,T;\mathcal{L}^2_{\lambda}(\R))$.

\begin{theorem}\label{ExistenceStrongSolution}
Under Assumption {A}, the variational inequality in \eqref{L_general_VI_strong} has a unique strong solution.
\end{theorem}
\begin{proof}
Assumption A is equivalent to assumptions (2.223), (2.224), (2.238), (2.239), (2.240)   of \citep[Chap. 3]{Bensoussan78}, and we also follow their Remark 2.24 to use  $\lambda(x)=e^{-n|x|}$ for some arbitrarily fixed  $n>0$ in our definition of Hilbert spaces. In turn, we can apply their Theorem 2.21 and our statement follows.
\end{proof}

Our main objective is to verify that  Assumption {A} is satisfied for our applications so that  Theorem \ref{ExistenceStrongSolution} applies  to ensure the existence of a unique strong solution to the VI \eqref{L_general_VI}.
To see this, we first write down  the operators associated with the log-price $X_t=\log(S_t)$ under   the GBM  and exponential OU models, namely, 
\begin{equation} \notag
\mathcal{A}[v]=\frac{\sigma^2}{2}\frac{\partial^2 v}{\partial x^2}+\mu\frac{\partial v}{\partial x}, \qquad
\mathcal{A}[v]=\frac{\sigma^2}{2}\frac{\partial^2 v}{\partial x^2}+(\hat{\theta}-\beta x)\frac{\partial v}{\partial x}. 
\end{equation} Therefore, $a_2 = \sigma^2/2$ is constant  and $a_1$ is an affine function in $x$  for both cases, so these coefficients  meet  the requirements in Assumption A. 

It remains  to verify that  the drive function $g(t,x) =G(t,e^x) \in \mathcal{H}^1(0,T;\mathcal{L}^2_{\lambda}(\R))$. In view of \eqref{H1norm},  we want to show that there exists $n>0$ such that
\begin{align*}
\int_0^T|\!|g(t,x)|\!|_{\mathcal{L}^2_{\lambda}(\R)}^2dt
	&=\int_0^T\int_{\R}\left(g(t,x)e^{-\frac{n}{2}|x|}\right)^2dx\,dt,  \quad \text{and}\\
\int_0^T|\!|\frac{\partial g}{\partial t}(t,x)|\!|_{\mathcal{L}^2_{\lambda}(\R)}^2dt
	&=\int_0^T\int_{\R}\left(\frac{\partial g}{\partial t}(t,x)e^{-\frac{n}{2}|x|}\right)^2dx\,dt \\
\end{align*}
are finite, where  
\begin{align*}
g(t,x)&= (r-\mu(t,e^x))e^xV_s(t,e^x)-\alpha\psi((m-V(t,e^x))^+),\\
\frac{\partial g}{\partial t}(t,x)&=(r-\mu(t,e^x))e^xV_{ts}(t,e^x)+\alpha\psi'((m-V(t,e^x))^+)V_t(t,e^x)\indic{m>V(t,e^x)}.
\end{align*} Here,   the subscripts of $V$ indicate the partial derivatives in $t$ and $s$.  Recall the drift functions  $\mu(t,e^x)=\mu$  under GBM and $\mu(t,e^x)=\beta(\theta-x)$ under exponential OU models. We notice that, in both cases,  the drift does not depend on $t$, so we  just write $\mu(e^x)$. Also, we observe that $\psi$ and $\psi'$ are increasing, and $\psi'(\ell)$ is bounded for any finite $\ell$.  For both call and put options,  there exist positive constants $h_1,\,q_1,\,h_2,\,q_2$ such that $|V_s(t,e^x)|\leq 1, |V_t(t,e^x)|\leq h_1e^x+q_1,\,|V_{st}(t,e^x)|\leq h_2e^x+q_2$. Together, these   imply the time-independent bounds for both models:
\begin{align*} 
|g(t,x)|&\leq |r-\mu(e^x)|e^x+\alpha\psi(m)=\emph{o}({e^{2|x|}}),  \\
|\frac{\partial g}{\partial t}(t,x)|&\leq |r-\mu(e^x)|(h_1e^x+q_1)e^x+\alpha\psi'(m)(h_2e^x+q_2)=\emph{o}(e^{2|x|}).
\end{align*} 
This implies that by choosing   $n>4$, we have
\begin{align*}
\int_0^T|\!|g(t,x)|\!|_{\mathcal{L}^2_{\lambda}(\R)}^2dt &\leq \int_0^T\big|\!\big| |r-\mu(e^x)|e^x+\alpha\psi(m)\big|\!\big|_{\mathcal{L}^2_{\lambda}(\R)}^2dt<\infty,\\
\int_0^T|\!|\frac{\partial g}{\partial t}(t,x)|\!|_{\mathcal{L}^2_{\lambda}(\R)}^2dt&\leq \int_0^T\big|\!\big| |r-\mu(e^x)|(h_1e^x+q_1)e^x+\alpha\psi'(m)(h_2e^x+q_2)\big|\!\big|_{\mathcal{L}^2_{\lambda}(\R)}^2dt<\infty.
\end{align*}
Hence, we conclude that  $g\in \mathcal{H}^1(0,T;\mathcal{L}^2_{\lambda}(\R))$ for both puts and calls under the GBM and exponential OU models, and Assumption A is satisfied. 

\vspace{5 mm}
As a final remark, Sect. 3.4 of  \cite{Bensoussan78} also provides the probabilistic representation of the  strong solution $u(t,x)$ of  the VI \eqref{u_general_VI}, given by \begin{equation} \label{log_prices}
u(t,x)=\sup_{\tau \in \mathcal{T}_{t,T}} \E_{t,x} \left\{ \int_t^\tau e^{-r(u-t)} g(u,X_u) \,du\right\}, 
\end{equation}
where $dX_u=\eta(u,X_u)du+\kappa(u,X_u)dW_u$ and $X_t=x$. By  the definition $L(t,e^x) = u(t,x)$,  the optimal stopping problem in \eqref{log_prices} resembles  that  for the   optimal liquidation premium in  \eqref{L_G}.

\appendix
\section{Novikov Condition} \label{appendix_A}
Under the GBM model, the Sharpe ratio $\lambda=\frac{\mu-r}{\sigma}$ is constant, so the Novikov condition is clearly met. Let us consider the exponential OU case. By Corollary 3.5.14 of  \cite{KaratzasShreve91}, it suffices to show  that for every $t\in[0,T]$,
\begin{equation} \label{Novikov}
\exists \, \varepsilon>0 \text{ s.t. } \E\left\{\exp\left(\int_t^{t+\varepsilon}\frac{1}{2}\lambda^2(u,S_u)du\right)\right\}<\infty.
\end{equation}
Using Jensen's inequality and Tonelli's Theorem, we have, for every $\varepsilon>0$,
\begin{align*}
\E\left\{\exp\left(\int_t^{t+\varepsilon}\frac{1}{2}\lambda(u,S_u)^2du\right)\right\}
& \leq 	\frac{1}{\varepsilon}\int_t^{t+\varepsilon}\E\left\{e^{\frac{\varepsilon}{2}\lambda(u,S_u)^2}\right\}du \\
& =		\frac{1}{\varepsilon}\int_t^{t+\varepsilon}\E\left\{e^{\frac{\varepsilon}{2}\left\{K_1[\theta-\log(S_u)]^2
		+K_2+K_3[\theta-\log(S_u)]\right\}}\right\}du.
\end{align*}
for some real constants $K_1$, $K_2$ and $K_3$. Moreover, Cauchy-Schwarz inequality implies that
\begin{equation}
\E\left\{e^{\frac{\varepsilon}{2}\left\{K_1[\theta-\log(S_u)]^2
		+K_2+K_3[\theta-\log(S_u)]\right\}}\right\}
 \leq 	\sqrt{\E\left\{e^{\varepsilon\left\{K_1[\theta-\log(S_u)]^2+K_2\right\}}\right\}
		\E\left\{e^{\varepsilon K_3(\theta-\log(S_u))}\right\}}. \label{ineq1}
\end{equation}
Since $\log(S_u)$ has the normal distribution, there exists $\varepsilon >0$ s.t. $\E\left\{\exp\left(\varepsilon K_1[\theta-\log(S_u)]^2+K_2\right)\right\}$ and $\E\left\{\exp\left(\varepsilon K_3[\theta-\log(S_u)]\right)\right\}$ in \eqref{ineq1} are finite. Hence, condition \eqref{Novikov} holds.

\begin{small}
\bibliographystyle{apa}
\singlespacing
\bibliography{mybib2}
\end{small}
\end{document}